\documentclass[11pt,a4paper]{article}
\usepackage[T1]{fontenc}
\usepackage{lmodern}
\usepackage[margin=25mm]{geometry}
\usepackage{amsmath,amssymb,amsthm,mathtools}
\usepackage{booktabs,array,microtype}
\usepackage[hidelinks]{hyperref}
\allowdisplaybreaks
\newtheorem{theorem}{Theorem}[section]
\newtheorem{proposition}[theorem]{Proposition}
\newtheorem{lemma}[theorem]{Lemma}
\newtheorem{corollary}[theorem]{Corollary}
\theoremstyle{definition}

\newtheorem{example}[theorem]{Example}
\theoremstyle{remark}
\newtheorem{remark}[theorem]{Remark}
\newcommand{\R}{\mathbb R}

\newcommand{\Z}{\mathbb Z}

\newcommand{\diag}{\operatorname{diag}}

\newcommand{\norm}[1]{\lVert#1\rVert}
\numberwithin{equation}{section}
\title{Moment Methods for Uniform Average Mixing\\
on Strongly Regular Graphs}
\author{Musung Kang}
\date{7 October 2026}
\begin{document}
\maketitle
\begin{abstract}
We study continuous-time quantum walks on connected strongly regular graphs that are not complete, observed at a random time drawn from a freely chosen probability law. Uniform average mixing (UAM) asks for a law under which every averaged transition probability equals $\displaystyle\frac{1}{n}$, where $n$ is the number of vertices. On a strongly regular graph this is equivalent to two affine constraints on three cosine moments. We construct a bounded, compactly supported time density for every strongly regular graph with nonintegral eigenvalues. For integral spectra we give an exact finite Toeplitz criterion and its Hankel form. Every averaged mixing matrix of such a graph is realized by at most two observation times.
Three elementary inequalities on the moment line, which also give a short proof of Chan's classification of complex Hadamard matrices in the Bose--Mesner algebra, lead to a determination of all strongly regular graphs that admit UAM. Apart from the conference graphs of nonsquare order and the graphs with instantaneous uniform mixing, these are the members of two infinite families of parameter sets and their complements, and for them we give explicit laws with two observation times. The Petersen graph and its complement are the smallest members. A strongly regular graph with UAM admits a bounded time density exactly when it has no instantaneous uniform mixing.
We also correct the classification of instantaneous uniform mixing on strongly regular graphs by Godsil, Mullin and Roy. Its sign condition excludes the halved $5$-cube, which mixes uniformly at time $\displaystyle\frac{\pi}{4}$. With the order $4\theta^2$ read literally, its parity condition also excludes the Clebsch graph and includes the parameters $(36,14,4,6)$, for which no time law gives uniform average mixing.
\end{abstract}

\medskip
\noindent\textbf{Keywords.} Quantum walk, uniform mixing, average mixing, strongly regular graph, moment problem, complex Hadamard matrix.

\noindent\textbf{Mathematics Subject Classifications.} 05E30, 05C50, 81P45.

\section{Scope and main results}
Let $G$ be a finite simple graph on $n$ vertices, with adjacency matrix $A$. Set
\[
 U(t)=e^{itA},\quad M(t)=U(t)\circ\overline{U(t)}.
\]
A Borel probability measure $\mu$ on $[0,\infty)$ is a \emph{UAM time law} if
\begin{equation}\label{eq:uam}
  \int_{0}^{\infty} M(t)\,d\mu(t)=\frac{1}{n}J.
\end{equation}
The same law must work for every starting vertex. This is averaging over a chosen law, not necessarily the uniform average over a long time interval. Since $M(-t)=M(t)$, allowing measures on $\R$ gives the same existence question. The general spectral-moment viewpoint is developed in \cite{BCM}, where a law is called a sampling distribution, and a spectral theory of such averages is posed as an open problem in \cite[Open Problem~4.6]{CG}. Approximate uniform mixing of discrete-time quantum walks on strongly regular graphs is studied in \cite{Zhan}. A point mass $\delta_t$ is a UAM time law exactly when every entry of $U(t)$ has modulus $\displaystyle\frac{1}{\sqrt n}$. This is \emph{instantaneous uniform mixing} (IUM) at time $t$.

Throughout Sections~\ref{sec:reduction}--\ref{sec:examples}, $G$ is connected and not complete, and it is strongly regular with parameters $(n,k,\lambda,\nu)$.
The symbol $\nu$ denotes the fourth graph parameter, while $\mu$ denotes a time law.
The two restricted eigenvalues are denoted by $r>s$, and $k$ is a simple eigenvalue because $G$ is connected.

We prove the following results.
\begin{enumerate}
  \item UAM is equivalent to the affine three-moment conditions of Theorem~\ref{thm:line}.
  \item Every such strongly regular graph with nonintegral restricted eigenvalues admits a bounded probability density supported in a finite interval $[0,T]$. Both the density and a sufficient $T$ are specified in Theorem~\ref{thm:conference}.
  \item For integral restricted eigenvalues, Theorem~\ref{thm:toeplitz} gives a necessary and sufficient finite semidefinite condition. Section~\ref{sec:hankel} gives the equivalent real Hankel condition and recovers observation times from its representing measure.
  \item Every averaged mixing matrix of a connected strongly regular graph that is not complete is realized by at most two observation times (Theorem~\ref{thm:twoatoms}), and the support interval can be retained. For nonintegral spectra exactly two atoms are necessary and sufficient (Corollary~\ref{cor:conference-twoatoms}).
  \item Theorem~\ref{thm:classification} determines all strongly regular graphs with UAM. Besides the conference graphs of nonsquare order and the graphs with IUM, they are the graphs whose parameters or complementary parameters are $(4\theta^2+4\theta+2,\,2\theta^2+\theta,\,\theta^2-1,\,\theta^2)$ with $\theta\ge1$, or $(4\theta^2-1,\,2\theta^2,\,\theta^2,\,\theta^2)$ with $\theta\ge2$. Propositions~\ref{prop:family-v} and~\ref{prop:family-iii} give explicit laws with two observation times, and Proposition~\ref{prop:no-single-time} shows that one time never suffices for them. By Proposition~\ref{prop:density} a graph with UAM admits a UAM law with a bounded density exactly when it has no IUM.
  \item Theorem~\ref{thm:gmr-corrected} corrects \cite[Theorem~5.1 and Lemma~5.2]{GMR}, and Example~\ref{ex:clebsch} gives counterexamples to the statements as printed.
\end{enumerate}
The key steps for the classification are Lemmas~\ref{lem:gram-hadamard} and~\ref{lem:three-ineq}. A positive semidefinite moment Gram matrix on the moment line exists exactly when the Bose--Mesner algebra of $G$ contains a complex Hadamard matrix. Three elementary inequalities then leave five parameter families (Proposition~\ref{prop:elementary-chan}). This recovers Chan's classification of such matrices \cite{Chan11}, in the form of \cite[Theorem~4.1]{GMR}, by a short argument, and we decide the five families one by one. The classification does not depend on the classification theorems of \cite{Chan11} and \cite{GMR}, and no proof uses computer calculation.

\section{Reduction to three moments}\label{sec:reduction}
Strong regularity is encoded in the identity
\begin{equation}\label{eq:srg}
 A^2=(\lambda-\nu)A+(k-\nu)I+\nu J.
\end{equation}
Consequently
\[
 r+s=\lambda-\nu,\quad rs=\nu-k,\quad
 \Delta:=r-s>0.
\]
The spectral projections are
\begin{equation}\label{eq:projectors}
 E_0=\frac{1}{n}J,\quad
 E_r=\frac{1}{\Delta}\left(A-sI-\frac{k-s}{n}J\right),\quad
 E_s=\frac{1}{\Delta}\left(rI-A+\frac{k-r}{n}J\right).
\end{equation}
They give
\begin{equation}\label{eq:mixingexpansion}
 M(t)= \sum_{j \in \{0,r,s\}} E_j \circ E_j
 +2\cos((k-r)t)\,E_0\circ E_r
 +2\cos((k-s)t)\,E_0\circ E_s
 +2\cos(\Delta t)\,E_r\circ E_s .
\end{equation}
Integrating \eqref{eq:mixingexpansion} against a time law gives the strongly regular case of \cite[Proposition~5]{BCM}.

Define the positive frequencies and their moments by
\begin{align}
 \omega_1&=k-r,&\omega_2&=k-s,&\omega_3&=r-s,
 &\omega_2&=\omega_1+\omega_3,  \notag\\
 x&= \int\cos(\omega_1t)\,d\mu(t),&
 y&= \int\cos(\omega_2t)\,d\mu(t),&
 z&= \int\cos(\omega_3t)\,d\mu(t).&\label{eq:xyz}
\end{align}
These are three labelled moments. If two frequencies coincide, their moments coincide as well.

For an adjacent pair of vertices $g,h$, the scaled entries $n(E_r)_{gh}$ and $n(E_s)_{gh}$ are
\[
 a=\frac{n-k+s}{\Delta},\quad -1-a.
\]
For a distinct nonadjacent pair they are
\[
 b=-\frac{k-s}{\Delta},\quad -1-b.
\]
In particular $\displaystyle a-b=\frac{n}{\Delta}\ne0$. For either $u=a$ or $u=b$, the averaged transition probability for the corresponding pair is
\begin{equation}\label{eq:entry}
 \frac{1+u^2+(1+u)^2+2ux-2(1+u)y-2u(1+u)z}{n^2}.
\end{equation}

\begin{theorem}[The SRG moment line]\label{thm:line}
Put
\begin{align}
 \kappa&=1-\frac{n}{2},\notag\\
 P&=\frac{(n-k+r)(k-r)}{\Delta^2}=-(a+1)(b+1),&
 Q&=\frac{(n-k+s)(k-s)}{\Delta^2}=-ab.\label{eq:PQ}
\end{align}
A probability law $\mu$ is a UAM time law if and only if its moments satisfy
\begin{equation}\label{eq:line}
 \boxed{\quad x=\kappa+P\eta,\quad y=\kappa+Q\eta,\quad z=1-\eta
 \quad\text{for some }\eta\in[0,2].\quad}
\end{equation}
Here $\eta=1-z$ is a moment parameter, not an observation time.
\end{theorem}
\begin{proof}
Divide the numerator of \eqref{eq:entry} by $2$ to obtain
\[
 1-y+(1-z)u(u+1)+u(x-y).
\]
Requiring \eqref{eq:entry} to equal $\displaystyle\frac{1}{n}$ for both $a$ and $b$ and subtracting gives
\[
 x-y=-(a+b+1)(1-z).
\]
Substitution into the equation for $a$ gives
\[
 y=1-\frac{n}{2}-ab(1-z),\quad
 x=1-\frac{n}{2}-(a+1)(b+1)(1-z),
\]
which is \eqref{eq:line}. The range of $\eta$ follows from $-1\le z\le1$.
Conversely, these identities make every off-diagonal entry of the averaged mixing matrix equal to $\displaystyle\frac{1}{n}$. Each row sums to one, so its diagonal entry is also $\displaystyle\frac{1}{n}$.
\end{proof}

Equivalently, put $R(t)=(R_1(t),R_2(t))\in\R^2$ with
\begin{equation}\label{eq:R}
 R_1(t)=\cos(\omega_1t)-\kappa-P(1-\cos(\omega_3t)),\quad
 R_2(t)=\cos(\omega_2t)-\kappa-Q(1-\cos(\omega_3t)).
\end{equation}
Then $\mu$ is a UAM time law if and only if $\displaystyle\int R\,d\mu=0$, and $\delta_t$ is one if and only if $R(t)=0$. Note that $\displaystyle R(0)=(1-\kappa)(1,1)=\frac{n}{2}(1,1)$.

\begin{corollary}[Parameter invariance]\label{cor:parameter}
Graphs with the same SRG parameters have exactly the same UAM time laws.
\end{corollary}
\begin{proof}
Formula \eqref{eq:entry} depends only on the parameters and the three moments.
\end{proof}

\subsection{A support bound with two observation times}\label{sec:twoatoms}
We record the continuous-curve form of Fenchel's refinement of
Carath\'eodory's theorem, see \cite[Section 1 and Theorem 1.5]{BK}.
An elementary proof is included to make the support reduction explicit.

\begin{lemma}[Convex combinations on a continuous curve]\label{lem:curve}
Let $d\ge1$ and let $F:[a,b]\to\R^d$ be continuous. Every point in
$\operatorname{conv}F([a,b])$ is a convex combination of at most $d$
points of $F([a,b])$.
\end{lemma}
\begin{proof}
Ordinary Carath\'eodory gives at most $d+1$ points. If $d$ points do
not already suffice, write the target $q$ as an interior point of a
full-dimensional simplex with vertices $x_0,\ldots,x_d$ on the curve.
There is a continuous path $\Gamma:[0,1]\to F([a,b])$ from $x_0$
to $x_1$, obtained by traversing the corresponding parameter interval.
Set
\[
 S_u=\operatorname{conv}\{\Gamma(u),x_1,\ldots,x_d\},\quad
 E=\{u\in[0,1]:q\in S_u\}.
\]
The set $E$ is compact, because along a convergent sequence in $E$ the
barycentric weights have a convergent subsequence in the standard
simplex. Also $0\in E$. If $1\in E$, the simplex $S_1$ has at most
$d$ distinct vertices and we are done. Otherwise let $u_*=\max E<1$.
If $S_{u_*}$ is full dimensional and $q$ lies in its interior, its
strictly positive barycentric coordinates persist for $u>u_*$
sufficiently close to $u_*$, a contradiction. Thus $q$ lies on a
proper face, or $S_{u_*}$ has dimension at most $d-1$. In either case
Carath\'eodory in that affine span expresses $q$ using at most $d$
of its vertices.
\end{proof}

\begin{lemma}[Atomic compression of continuous observables]\label{lem:observable-compression}
Let $F:[0,\infty)\to\R^d$ be bounded and continuous, with $d\ge1$,
and let $\mu$ be a Borel probability law. There are at most $d$
times and nonnegative weights of sum one such that
\[
 \int F(t)\,d\mu(t)=\sum_{j=1}^{\ell}p_jF(t_j),\quad \ell\le d.
\]
If $\mu$ is concentrated on a compact interval $[a,b]\subset[0,\infty)$,
all the $t_j$ can be chosen in that interval.
\end{lemma}
\begin{proof}
Let $\mu$ be concentrated on a Borel set $X\subseteq[0,\infty)$, and put $\displaystyle q=\int F\,d\mu$ and $K=\operatorname{conv}F(X)$. We show $q\in K$ by induction on the dimension of the affine hull of $F(X)$, the case of a single point being trivial. The point $q$ lies in this affine hull, and it lies in the closure of $K$ because $\displaystyle\ell(q)=\int\ell(F)\,d\mu\le\sup_K\ell$ for every linear functional $\ell$. If $q$ is not in the relative interior of $K$, a supporting hyperplane gives a linear functional $\ell$, not constant on $F(X)$, with $\ell\le\ell(q)$ on $K$. Then $\displaystyle\int\left(\ell(q)-\ell(F)\right)d\mu=0$ with a nonnegative integrand, so $\mu$ is concentrated on $X'=\{t\in X:\ell(F(t))=\ell(q)\}$, and the affine hull of $F(X')$ has smaller dimension. By induction $q\in\operatorname{conv}F(X')\subseteq K$. This is the elementary fact behind the atomic representation theorem \cite[Corollary 2]{BT}.

By Carath\'eodory's theorem $q$ is a convex combination of at most $d+1$ points $F(t_j)$ with $t_j\in X$. These times lie in a compact interval, which can be taken to be $[a,b]$ when $\mu$ is concentrated on $[a,b]$. Lemma~\ref{lem:curve}, applied to $F$ on that interval, reduces the number of points to $d$. In particular $q$ lies in the convex hull itself, not merely in its closure, and no assumption on rational relations among the frequencies is used.
\end{proof}

\begin{theorem}[Two-time realization of every SRG average]\label{thm:twoatoms}
Let $G$ be a connected strongly regular graph that is not complete, and let $\mu$ be any Borel probability
law on $[0,\infty)$. Then there exist $t_1,t_2\ge0$ and $p\in[0,1]$
such that
\begin{equation}\label{eq:twoaverage}
 \boxed{\quad
 \int M(t)\,d\mu(t)=pM(t_1)+(1-p)M(t_2).
 \quad}
\end{equation}
If $\mu$ is concentrated on a compact interval $[a,b]$, both times can
be chosen in $[a,b]$. In particular,
\begin{equation}\label{eq:twouam}
 \boxed{\quad
 G\text{ admits UAM}
 \quad\Longleftrightarrow\quad
 \exists\,p\in[0,1],\ t_1,t_2\ge0\ \text{ with }
 \ pM(t_1)+(1-p)M(t_2)=\frac{1}{n}J .
 \quad}
\end{equation}
\end{theorem}
\begin{proof}
By \eqref{eq:srg}, the space spanned by $I,A,J$ is an algebra under
ordinary matrix multiplication. Its basis $I,A,J-I-A$ has disjoint
entry supports, so the algebra is also closed under Schur products
and entrywise conjugation. Hence $U(t)$ and $M(t)$ belong to it, and
there are bounded continuous real functions $\beta_0,\beta_1,\beta_2$ with
\begin{equation}\label{eq:threeclasses}
 M(t)=\beta_0(t)I+\beta_1(t)A+\beta_2(t)(J-I-A).
\end{equation}
These functions are the diagonal, adjacent-pair, and nonadjacent-pair
transition probabilities. Stochasticity gives
\begin{equation}\label{eq:rowconstraint}
 \beta_0(t)+k\beta_1(t)+(n-k-1)\beta_2(t)=1.
\end{equation}
Apply Lemma~\ref{lem:observable-compression} with $d=2$ and
$F(t)=(\beta_1(t),\beta_2(t))$. The two resulting times preserve both
off-diagonal class averages, and \eqref{eq:rowconstraint} preserves the
diagonal average as well. This proves \eqref{eq:twoaverage} and the
interval assertion. Equation \eqref{eq:twouam} follows, with the
reverse implication given by the probability law
$p\delta_{t_1}+(1-p)\delta_{t_2}$.
\end{proof}

\begin{corollary}[An exact two-time test]\label{cor:two-time-test}
Put
\[
 \Psi(t)=\left(\beta_1(t)-\frac{1}{n},\ \beta_2(t)-\frac{1}{n}\right)\in\R^2 .
\]
Then $G$ admits UAM if and only if either $\Psi(t)=0$ at some time, or
there are $t_1,t_2\ge0$ for which
\begin{equation}\label{eq:opposite}
 \det[\Psi(t_1),\Psi(t_2)]=0,\quad
 \Psi(t_1)\mathbin{\cdot}\Psi(t_2)<0 .
\end{equation}
In the second case an exact UAM law is
\begin{equation}\label{eq:two-weights}
 \mu=p\delta_{t_1}+(1-p)\delta_{t_2},\quad
 p=\frac{\norm{\Psi(t_2)}_2}
         {\norm{\Psi(t_1)}_2+\norm{\Psi(t_2)}_2}.
\end{equation}
\end{corollary}
\begin{proof}
The first case is instantaneous uniform mixing. Otherwise
\eqref{eq:twouam} holds precisely when two nonzero vectors $\Psi(t_1)$
and $\Psi(t_2)$ are opposite positive multiples. In two dimensions this
is equivalent to \eqref{eq:opposite}, and \eqref{eq:two-weights}
is their cancelling convex combination.
\end{proof}

\begin{remark}[What the support bound preserves]\label{rem:two-scope}
Theorem~\ref{thm:twoatoms} preserves the entire averaged mixing matrix.
It need not preserve all three labelled cosine moments $x,y,z$,
since their image under \eqref{eq:mixingexpansion} has only two independent
affine coordinates. If those three moments themselves are prescribed,
Lemma~\ref{lem:observable-compression} with
$F(t)=(\cos(\omega_1t),\cos(\omega_2t),\cos(\omega_3t))$ gives
at most three times, again retaining a compact support interval.
The two-time theorem is an existence reduction, not a uniform upper
bound on the times. It applies to arbitrary Borel time laws, and it does not
assert two-atom realization inside the subclass of densities.
The same proof shows that every averaged mixing matrix of a graph whose adjacency matrix lies in the Bose--Mesner algebra of a symmetric association scheme with $d$ classes, for instance a distance-regular graph of diameter $d$, is realized by at most $d$ observation times.
\end{remark}

\begin{remark}\label{rem:gram}
Every time law necessarily satisfies
\[
 C(x,y,z)=
 \begin{pmatrix}1&x&y\\x&1&z\\y&z&1\end{pmatrix}\succeq0.
\]
Indeed, this is the real part of the Gram matrix of the three phase functions
$1,e^{i\omega_1t},e^{i\omega_2t}$ in $L^2(\mu)$. Equivalently,
\[
 |x|,|y|,|z|\le1,\quad
 1+2xyz-x^2-y^2-z^2\ge0.
\]
This is only a necessary condition for a time law. Arithmetic relations among the frequencies can impose further moment constraints. The example with parameters $(25,12,5,6)$ in Section~\ref{sec:examples} passes this small Gram test on an interval of the moment line and nevertheless has no UAM law. For primitive graphs, Lemma~\ref{lem:three-ineq} and Proposition~\ref{prop:elementary-chan} show that this necessary condition alone already restricts the parameters to five families.
\end{remark}

\section{An explicit time density for nonintegral spectra}\label{sec:conference}
\begin{lemma}[The nonintegral parameters]\label{lem:conference}
If the restricted eigenvalues are nonintegral, then $n$ is not a square and
\begin{equation}\label{eq:conference}
 k=\frac{n-1}{2},\quad \lambda=\frac{n-5}{4},\quad
 \nu=\frac{n-1}{4},\quad
 r=\frac{-1+\sqrt n}{2},\quad s=\frac{-1-\sqrt n}{2}.
\end{equation}
The multiplicities of $r$ and $s$ are both $\displaystyle\frac{n-1}{2}$.
\end{lemma}
\begin{proof}
The eigenvalues are roots of the monic integer quadratic supplied by \eqref{eq:srg}. A rational root is an integer, so nonintegral roots are irrational conjugates. If their multiplicities are $f$ and $g$, the irrational part of $k+fr+gs=0$ gives $\displaystyle f=g=\frac{n-1}{2}$. Thus
\[
 r+s=-\frac{2k}{n-1}.
\]
The left side is an integer and $0<k<n-1$, so it equals $-1$ and $\displaystyle k=\frac{n-1}{2}$. The identity $\operatorname{tr}(A^2)=nk$, together with $r+s=-1$ and $rs=\nu-k$, gives $\displaystyle\nu=\frac{k}{2}$. Then $\lambda=\nu-1$ and $(r-s)^2=n$, proving the formulas. Irrationality implies that $n$ is not a square.
\end{proof}

These are the usual conference parameters, and background on the spectra of strongly regular graphs is given in \cite{Cameron,BH}. The proof above records all facts needed here.

For \eqref{eq:conference}, the moment line reduces to
\begin{equation}\label{eq:conference-line}
 x=y,\quad 2x+\frac{n-1}{2}z=\frac{3-n}{2}.
\end{equation}
In particular, the target
\begin{equation}\label{eq:conference-target}
 x=0,\quad y=0,\quad z=z_*:=-\frac{n-3}{n-1}
\end{equation}
is sufficient.

Let
\begin{align}
 \gamma&=\sqrt n,&
 \omega_1&=\frac{n-\sqrt n}{2},&
 \omega_2&=\frac{n+\sqrt n}{2},&
 \omega_3&=\gamma,\notag\\
 N&=n-3,&
 \varepsilon&=\frac{1}{n-1},&
 \alpha&=1-\varepsilon=\frac{n-2}{n-1}.\label{eq:constants}
\end{align}
The Fej\'er polynomial is
\begin{equation}\label{eq:fejer}
 F_N(u)=\frac{1}{N+1}\left|\sum_{\ell=0}^{N}e^{i\ell u}\right|^2
 =1+2\sum_{\ell=1}^{N}\left(1-\frac{\ell}{N+1}\right)\cos(\ell u)\ge0.
\end{equation}
Define the positive bounded function
\begin{align}
 W(t)&=\varepsilon+\alpha F_N(\gamma t-\pi)
      =1+2\sum_{\ell=1}^{N}b_\ell\cos(\ell\gamma t),\label{eq:W}\\
 b_\ell&=\alpha\left(1-\frac{\ell}{N+1}\right)(-1)^\ell,\quad
 W(t)\ge\varepsilon.\notag
\end{align}
The function $W$ is not by itself a probability density on the whole half-line.
It will be restricted to a finite interval and corrected.

Since $\displaystyle\frac{\omega_1}{\gamma}=\frac{\sqrt n-1}{2}$ and $\displaystyle\frac{\omega_2}{\gamma}=\frac{\sqrt n+1}{2}$ are irrational, neither is an integer. Elementary averaging of the finitely many cosine terms gives
\begin{equation}\label{eq:means}
 \lim_{T\to\infty}\frac{1}{T}\int_0^T
 W(t)\begin{pmatrix}1\\\cos(\omega_1t)\\\cos(\omega_2t)\\\cos(\gamma t)\end{pmatrix}dt
 =
 \mathbf d:=\begin{pmatrix}1\\0\\0\\z_*\end{pmatrix}.
\end{equation}
The last coordinate is $\displaystyle b_1=-\frac{\alpha N}{N+1}=z_*$. Thus $W$ supplies the desired limiting moments with a positive margin. The next construction corrects the finite-interval errors exactly.

Set $\omega_0=0$, and let $\mathcal F$ be the finite set of the \emph{nonzero} numbers among
\[
 \left\{|\omega_i+\omega_j|,\ |\omega_i-\omega_j|\,\middle|\,0\le i,j\le3\right\}
 \ \cup\
 \left\{|\omega_i+\ell\gamma|,\ |\omega_i-\ell\gamma|\,\middle|\,
        0\le i\le3,\ 1\le\ell\le N\right\}.
\]
Define
\begin{equation}\label{eq:rho}
 \rho=\min\mathcal F>0,\quad
 B=1+2\sum_{\ell=1}^{N}|b_\ell|
   =1+\alpha N=\frac{n^2-4n+5}{n-1}.
\end{equation}
Zero frequencies are excluded from $\mathcal F$ because their averages are exactly one, not errors.

\begin{theorem}[A time density for nonintegral spectra]\label{thm:conference}
For every connected strongly regular graph that is not complete and has nonintegral restricted eigenvalues, choose
\begin{equation}\label{eq:Tbound}
 T\ge\frac{32(n-1)B}{\rho}
   =\frac{32(n^2-4n+5)}{\rho}.
\end{equation}
Let
\[
 f(t)=\begin{pmatrix}1\\\cos(\omega_1t)\\\cos(\omega_2t)\\\cos(\omega_3t)\end{pmatrix},
 \quad
 G_T=\frac{1}{T}\int_0^T f(t)f(t)^{\mathsf T}\,dt,\quad
 a_T=\frac{1}{T}\int_0^T W(t)f(t)\,dt
\]
and put
\begin{equation}\label{eq:correction}
 q_T=G_T^{-1}(\mathbf d-a_T).
\end{equation}
Then
\begin{equation}\label{eq:timedensity}
 \boxed{\quad
 d\mu_T(t)=\frac{\mathbf1_{[0,T]}(t)}{T}
          \left(W(t)+q_T^{\mathsf T}f(t)\right)\,dt
 \quad}
\end{equation}
is a probability law with \eqref{eq:conference-target}, and hence is a UAM time law.
Its density is bounded and is at least $\displaystyle\frac{1}{2T(n-1)}$ on $[0,T]$.
\end{theorem}
\begin{proof}
The three frequencies are distinct and positive. Indeed $\omega_1=\omega_2$ is impossible, $\omega_2=\gamma$ would force $n=1$, and $\omega_1=\gamma$ would force $n=9$, which is excluded because $n$ is not a square. The functions comprising $f$ are linearly independent on every interval, so their Gram matrix $G_T$ is positive definite.

Write $\displaystyle G_\infty=\diag\left(1,\frac{1}{2},\frac{1}{2},\frac{1}{2}\right)$. For every nonzero frequency $\xi\in\mathcal F$,
\[
 \left|\frac{1}{T}\int_0^T\cos(\xi t)\,dt\right|\le\frac{1}{T\rho}.
\]
The cosine product formula therefore gives
\[
 \|G_T-G_\infty\|_\infty\le\frac{4}{T\rho},\quad
 \|a_T-\mathbf d\|_\infty\le\frac{B}{T\rho}.
\]
Here the matrix infinity norm is the maximum absolute row sum. The second inequality counts the absolute Fourier coefficients of $W$, and the zero-frequency contributions are precisely $\mathbf d$ by \eqref{eq:means}.

Since $\|G_\infty^{-1}\|_\infty=2$, the Neumann-series bound gives
\[
 T\rho\ge16\quad\Longrightarrow\quad \|G_T^{-1}\|_\infty\le4.
\]
The stronger assumption \eqref{eq:Tbound} consequently implies
\[
 \|q_T\|_1\le4\|q_T\|_\infty
 \le\frac{16B}{T\rho}\le\frac{1}{2(n-1)}=\frac{\varepsilon}{2}.
\]
As every coordinate of $f$ has absolute value at most one,
\[
 W(t)+q_T^{\mathsf T}f(t)\ge\varepsilon-\|q_T\|_1\ge\frac{\varepsilon}{2}.
\]
Finally,
\[
 \int f(t)\,d\mu_T(t)=a_T+G_Tq_T=\mathbf d.
\]
The first coordinate proves normalization, and the remaining coordinates are exactly the required moments. Theorem~\ref{thm:line} proves UAM.
\end{proof}

The number $\rho$ is bounded below uniformly. Each element of $\mathcal F$ equals $\gamma|u|$, where either $|u|\ge1$ or $\displaystyle u=\pm\frac{\gamma-o}{2}$ with an odd integer $o$. In the second case $|\gamma-o|\ge1$, or else $o>0$ and $\displaystyle|\gamma-o|=\frac{|n-o^2|}{\gamma+o}\ge\frac{4}{2\gamma+1}$, because $n-o^2$ is a nonzero multiple of $4$ as $n\equiv1\pmod4$ is not a square. Hence $\displaystyle\rho\ge\frac{2\gamma}{2\gamma+1}>\frac{4}{5}$, and $T=40(n^2-4n+5)$ always satisfies \eqref{eq:Tbound}. Remark~\ref{rem:conference-explicit} below gives an explicit law with two atoms.

\subsection{All entries are elementary finite expressions}
No torus sampling or limiting measure is needed to evaluate \eqref{eq:timedensity}.
For
\[
 S_T(u)=
 \begin{cases}\dfrac{\sin(uT)}{uT},&u\ne0,\\[1ex] 1,&u=0,\end{cases}
\]
the matrix and vector in the theorem are
\begin{align}
 (G_T)_{ij}
 &=\frac{S_T(\omega_i-\omega_j)+S_T(\omega_i+\omega_j)}{2},
 &&0\le i,j\le3,\label{eq:Gramformula}\\
 (a_T)_i
 &=S_T(\omega_i)+\sum_{\ell=1}^{N}b_\ell
   \left(S_T(\omega_i-\ell\gamma)+S_T(\omega_i+\ell\gamma)\right),
 &&0\le i\le3.\label{eq:aformula}
\end{align}
Thus the construction consists of a finite minimum of explicit frequencies,
elementary sine evaluations, and one $4$ by $4$ linear solve.
The bound on $T$ is sufficient, and no optimality is asserted.

\begin{lemma}[No single observation time]\label{lem:no-single-time}
For the parameters \eqref{eq:conference} with $n$ not a square, no point mass $\delta_t$ is a UAM time law.
\end{lemma}
\begin{proof}
Suppose $\delta_t$ is a UAM time law. By \eqref{eq:conference-line}, $\cos(\omega_1t)=\cos(\omega_2t)$. As $\omega_2-\omega_1=\sqrt n$ and $\omega_1+\omega_2=n$, either $\sqrt n\,t\in2\pi\Z$ or $nt\in2\pi\Z$. In the first case $z=1$, and \eqref{eq:conference-line} gives $\displaystyle x=1-\frac{n}{2}<-1$, which is impossible. In the second case $nt=2\pi m$ with an integer $m\ge1$. Put $\displaystyle\varphi=\frac{\pi m}{\sqrt n}$. Then $\sqrt n\,t=2\varphi$ and $\omega_1t=\pi m-\varphi$, so $z=2\cos^2\varphi-1$ and $x=(-1)^m\cos\varphi$. Substituting into \eqref{eq:conference-line} gives
\[
 (n-1)\cos^2\varphi+2(-1)^m\cos\varphi-1=0.
\]
Hence $\cos\varphi$ is algebraic, and so is $e^{i\varphi}$, a root of $X^2-2\cos\varphi\,X+1$. On the other hand $\displaystyle\frac{m}{\sqrt n}$ is algebraic and irrational, and $e^{i\pi}=-1$ is algebraic. The Gelfond--Schneider theorem, in the form \cite[Theorem~3.3]{GMR}, therefore shows that $\displaystyle e^{i\varphi}=\exp\left(\frac{m}{\sqrt n}\,i\pi\right)$ is transcendental. This contradiction proves the lemma.
\end{proof}

\begin{corollary}[Exactly two times for nonintegral SRGs]\label{cor:conference-twoatoms}
Every connected strongly regular graph that is not complete and has nonintegral restricted eigenvalues
admits a UAM law with two distinct observation times and two strictly
positive weights. For every $T$ satisfying \eqref{eq:Tbound}, both times
can be chosen in $[0,T]$. One observation time never suffices.
\end{corollary}
\begin{proof}
Apply Theorem~\ref{thm:twoatoms} to the density $\mu_T$ of
Theorem~\ref{thm:conference}. It gives at most two times in $[0,T]$.
By Lemma~\ref{lem:conference} the graph has conference parameters with $n$ not a square, and by Lemma~\ref{lem:no-single-time} no point mass is a UAM law. The two times must therefore be distinct,
with $0<p<1$.
\end{proof}

\begin{remark}[An explicit two-atom law]\label{rem:conference-explicit}
For the parameters \eqref{eq:conference} with $n$ not a square and an integer $m\ge1$, put $\displaystyle t_m=\frac{2\pi m}{n}$ and $\displaystyle X_m=(-1)^m\cos\frac{\pi m}{\sqrt n}$. Here $\displaystyle P=Q=\frac{n-1}{4}$, and the computation in the proof of Lemma~\ref{lem:no-single-time} gives
\[
 R(t_m)=g_m\,(1,1),\qquad g_m=\frac{(n-1)X_m^2+2X_m-1}{2}.
\]
The quadratic $(n-1)X^2+2X-1$ is negative exactly for $\displaystyle-\frac{1}{\sqrt n-1}<X<\frac{1}{\sqrt n+1}$. Since $\displaystyle X_m=\cos\left(\pi m\left(1+\frac{1}{\sqrt n}\right)\right)$ and $\displaystyle1+\frac{1}{\sqrt n}$ is irrational, the numbers $X_m$ are dense in $[-1,1]$, so $g_m<0$ for some $m$. As $\displaystyle R(0)=\frac{n}{2}(1,1)$, the law
\[
 \mu=p\,\delta_0+(1-p)\,\delta_{t_m},\qquad p=\frac{-2g_m}{n-2g_m},
\]
is then a UAM time law with two atoms.
\end{remark}

\begin{remark}[Matrix preservation versus moment preservation]
The preceding corollary preserves UAM, but need not retain the
particular triple $(0,0,z_*)$ in \eqref{eq:conference-target}.
If that triple is also required, Lemma~\ref{lem:observable-compression}
gives at most three times in the same interval $[0,T]$.
The density \eqref{eq:timedensity} remains an explicit construction.
The two-time reduction is existential and supplies an exact support
bound, without an explicit formula for the two times.
\end{remark}

\section{An exact finite criterion for integral spectra}\label{sec:integral}
Assume that $r,s$ are integers. Let
\begin{equation}\label{eq:indices}
 \delta=\gcd(k-r,r-s),\quad
 a_0=\frac{k-r}{\delta},\quad
 b_0=\frac{r-s}{\delta},\quad D=a_0+b_0=\frac{k-s}{\delta}.
\end{equation}
Both $a_0,b_0$ are positive integers. The subscript distinguishes these indices from the entry parameters $a,b$ used in Section~\ref{sec:reduction}.

For real numbers $c_0,\ldots,c_D$, use the real symmetric Toeplitz matrix
\[
 \mathcal T_D(c)=(c_{|i-j|})_{i,j=0}^{D}.
\]
\begin{theorem}[Exact integral-spectrum criterion]\label{thm:toeplitz}
A connected strongly regular graph that is not complete and has integral spectrum admits UAM if and only if there exist
$c_0,\ldots,c_D\in\R$ and $\eta\in[0,2]$ such that
\begin{equation}\label{eq:SDP}
 \boxed{\begin{gathered}
 c_0=1,\quad \mathcal T_D(c)\succeq0,\quad 0\le\eta\le2,\\
 c_{a_0}=\kappa+P\eta,\quad
 c_{a_0+b_0}=\kappa+Q\eta,\quad
 c_{b_0}=1-\eta .
 \end{gathered}}
\end{equation}
Repeated indices denote the same variable and must satisfy all displayed equations.
The constants $\kappa,P,Q$ are those of \eqref{eq:PQ}.
\end{theorem}
\begin{proof}
A time law yields the circle variable $\phi=\delta t$ modulo $2\pi$, and
$\displaystyle c_j=\int\cos(j\phi)\,d\mu$. For every complex vector $u$,
\[
 u^*\mathcal T_D(c)u
 =\int\left|\sum_{j=0}^{D}u_je^{ij\phi}\right|^2\,d\tilde\mu(\phi)\ge0,
\]
where $\tilde\mu$ is the symmetrization of the pushed-forward circle law. Symmetrization preserves all cosine moments and makes its Fourier moments real. The affine equations follow from Theorem~\ref{thm:line}.

Conversely, the truncated trigonometric moment theorem
\cite[Section~11.4]{Sch} represents a real Toeplitz sequence with $c_0=1$ and
$\mathcal T_D(c)\succeq0$ by a circle probability measure. Symmetrize it if necessary and push it to $[0,\pi]$ by $\phi\mapsto\arccos(\cos\phi)$. All integer cosine moments are preserved. The map $\displaystyle t=\frac{\phi}{\delta}$ then gives a genuine time law on $\displaystyle\left[0,\frac{\pi}{\delta}\right]$ with the prescribed three moments. Apply Theorem~\ref{thm:line}.
\end{proof}

This is a finite semidefinite feasibility problem with affine constraints. It is an exact mathematical criterion, and a numerical approximate feasible matrix is not by itself an exact certificate. In this integral case the phase orbit is periodic, so there is no distinction between a circle measure and a realizable time law, even on the boundary of the moment cone.

\subsection{Equivalent Hankel conditions and recovery of a measure}\label{sec:hankel}
Let $\mathsf T_j$ denote the Chebyshev polynomial with
$\mathsf T_j(\cos\phi)=\cos(j\phi)$. Seek a truncated real moment sequence
$m_0,\ldots,m_D$ with $m_0=1$, and define $L_m(x^j)=m_j$.
Replace the last three equations of \eqref{eq:SDP} by
\begin{equation}\label{eq:chebyshev}
 L_m(\mathsf T_{a_0})=\kappa+P\eta,\quad
 L_m(\mathsf T_{a_0+b_0})=\kappa+Q\eta,\quad
 L_m(\mathsf T_{b_0})=1-\eta.
\end{equation}
The Toeplitz condition is equivalently replaced by the following pair of Hankel conditions.

If $D=2h$, require
\begin{equation}\label{eq:Hankel-even}
 H_h=(m_{i+j})_{i,j=0}^{h}\succeq0,\quad
 K_{h-1}=(m_{i+j}-m_{i+j+2})_{i,j=0}^{h-1}\succeq0.
\end{equation}
If $D=2h+1$, require
\begin{equation}\label{eq:Hankel-odd}
 H_h^+=(m_{i+j}+m_{i+j+1})_{i,j=0}^{h}\succeq0,\quad
 H_h^-=(m_{i+j}-m_{i+j+1})_{i,j=0}^{h}\succeq0.
\end{equation}
These are precisely the truncated $[-1,1]$ moment conditions of
\cite[Section~10.1]{Sch}.

For completeness, necessity follows by integrating $p(x)^2$ and
$(1-x^2)q(x)^2$ in the even case, and $(1+x)p(x)^2$ and
$(1-x)q(x)^2$ in the odd case. Sufficiency follows from the corresponding
representations of every nonnegative polynomial on $[-1,1]$ as sums of these terms, followed by the compact moment representation theorem.

Given a representing measure
\[
 \varrho=\sum_{\ell=1}^{L}p_\ell\delta_{\xi_\ell},\quad
 -1\le\xi_\ell\le1,\quad p_\ell\ge0,\quad \sum_\ell p_\ell=1,
\]
the observation-time law is
\begin{equation}\label{eq:recover}
 \boxed{\quad
 \mu=\sum_{\ell=1}^{L}p_\ell\,\delta_{t_\ell},\qquad t_\ell=\frac{\arccos\xi_\ell}{\delta}.
 \quad}
\end{equation}
Indeed $\cos(j\delta t_\ell)=\mathsf T_j(\xi_\ell)$.
A finite $\varrho$ always exists. The resulting UAM law can subsequently be compressed to at most two atoms in $\displaystyle\left[0,\frac{\pi}{\delta}\right]$ by Theorem~\ref{thm:twoatoms}. This preserves UAM, though not necessarily the entire completed sequence $c_0,\ldots,c_D$.

\begin{remark}[Semidefiniteness is sufficient here]
No strict positive definiteness is required in \eqref{eq:SDP},
\eqref{eq:Hankel-even}, or \eqref{eq:Hankel-odd}. Requiring it would lose valid boundary laws. Also, one Hankel matrix alone is not sufficient for support in $[-1,1]$, and the second, localizing matrix must be retained.
\end{remark}

\section{Examples and Obstructions}\label{sec:examples}
\subsection{A two-atom law for the Petersen graph}
For the parameters $(10,3,0,1)$,
\[
 (k,r,s)=(3,1,-2),\quad
 (\omega_1,\omega_2,\omega_3)=(2,5,3),\quad
 (\kappa,P,Q)=\left(-4,\frac{16}{9},\frac{25}{9}\right).
\]
Choose $\displaystyle\eta=\frac{9}{5}$. The required moments are
\[
 x=-\frac{4}{5},\quad y=1,\quad z=-\frac{4}{5}.
\]
Define
\begin{equation}\label{eq:petersen}
 \boxed{\quad
 \mu=\frac{9\sqrt5-20}{25}\,\delta_0
      +\frac{9(5-\sqrt5)}{25}\,\delta_{t_1},\qquad t_1=\frac{2\pi}{5}.
 \quad}
\end{equation}
Both weights are positive and sum to one. At $t=t_1$,
\[
 \cos(2t)=\cos(3t)=-\frac{1+\sqrt5}{4},\quad \cos(5t)=1.
\]
Writing $\displaystyle p=\frac{9(5-\sqrt5)}{25}$,
\[
 1-p+p\left(-\frac{1+\sqrt5}{4}\right)
 =1-\frac{p(5+\sqrt5)}{4}
 =-\frac{4}{5}.
\]
Thus \eqref{eq:petersen} has exactly the required moments and proves UAM.

One observation time does not suffice. For a point mass the matrix $C(x,y,z)$ of Remark~\ref{rem:gram} is the Gram matrix of three unit vectors in the plane, so $\det C=0$. On the moment line this happens only at $z=1$, which is excluded because $x=\kappa<-1$ there, at $\displaystyle z=-\frac{4}{5}$, where $y=1$, and at $\displaystyle z=-\frac{11}{16}$, where $x=-1$. If $y=\cos(5t)=1$, then $5t\in2\pi\Z$ and $\displaystyle\cos(3t)\in\left\{1,\frac{-1\pm\sqrt5}{4}\right\}$, which never equals $\displaystyle-\frac{4}{5}$. If $x=\cos(2t)=-1$, then $\cos(3t)=0$. Hence the Petersen graph admits a two-atom UAM law but no one-atom law. It is the case $\theta=1$ of Propositions~\ref{prop:family-v} and~\ref{prop:no-single-time}, which treat a whole family.

\subsection{A boundary law for the parameters $(9,4,1,2)$}
Here $(k,r,s)=(4,1,-2)$ and the frequencies are $(3,6,3)$.
The repeated frequency forces $x=z$. The moment line gives
\[
 x=y=-\frac{7}{2}+2\eta,\quad z=1-\eta,
\]
so $\displaystyle\eta=\frac{3}{2}$ and $\displaystyle x=y=z=-\frac{1}{2}$. Consequently
\begin{equation}\label{eq:nine}
 \mu=\delta_{t_2},\qquad t_2=\frac{2\pi}{9},
\end{equation}
is a UAM law, and in fact it gives instantaneous uniform mixing.

The Toeplitz sequence $\displaystyle c_0=1,\ c_1=c_2=-\frac{1}{2}$ has a singular positive semidefinite
$3$ by $3$ Toeplitz matrix. This is a useful test of the boundary case in
Theorem~\ref{thm:toeplitz}. Moreover every UAM law in this example satisfies
\[
 \int (1+2\cos(3t))^2\,d\mu(t)=3+4c_1+2c_2=0.
\]
Its support is therefore contained in the discrete set where $\displaystyle\cos(3t)=-\frac{1}{2}$.
In particular, no absolutely continuous probability law can give UAM for these parameters. The compactly supported density conclusion of Theorem~\ref{thm:conference} is specifically for nonintegral spectra.

\subsection{An elementary obstruction for the parameters $(25,12,5,6)$}
Every graph with these parameters has $(k,r,s)=(12,2,-3)$ and frequencies
$(10,15,5)$. Put $\phi=5t$ and $\displaystyle c_j=\int\cos(j\phi)\,d\mu$.
The moment line, with $P=Q=6$, requires
\[
 c_1=1-\eta,\quad c_2=c_3=-\frac{23}{2}+6\eta.
\]
In particular,
\begin{equation}\label{eq:25relation}
 c_2+6c_1+\frac{11}{2}=0.
\end{equation}
But pointwise,
\[
 \cos(2\phi)+6\cos\phi+\frac{11}{2}
 =2\left(\cos\phi+\frac{3}{2}\right)^2\ge\frac{1}{2}.
\]
Its integral cannot be zero. Therefore no graph with parameters $(25,12,5,6)$ admits UAM.

Equivalently, the first two monomial moments would satisfy
\[
 m_2-m_1^2
 =\frac{1+c_2}{2}-c_1^2
 =-\left(c_1+\frac{3}{2}\right)^2<0,
\]
contradicting a $2$ by $2$ Hankel condition. This is an exact moment obstruction, not a numerical infeasibility report.

Notice that the affine moment line does contain
\[
 (x,y,z)=\left(0,0,-\frac{11}{12}\right).
\]
The matrix $C(x,y,z)$ of Remark~\ref{rem:gram} is positive definite at this point.
It nevertheless cannot be a time-moment vector for frequencies $(10,15,5)$.
Thus checking only the Gram matrix indexed by the three eigenvalues does not solve the SRG problem. This is the case $\theta=2$ of Proposition~\ref{prop:square-conference}, which excludes every square order except $n=9$.

\section{Degenerate and imprimitive cases}\label{sec:edge}
A disconnected graph on more than one component cannot have UAM on all its vertices, because transition probabilities between components vanish for every time.

For the complete graph $K_n$, $n\ge2$, each off-diagonal mixing entry is
\[
 \frac{2(1-\cos(nt))}{n^2}.
\]
Hence UAM requires $\displaystyle\int\cos(nt)\,d\mu=1-\frac{n}{2}$, which is possible exactly for
$n\le4$. For $2\le n\le4$ a point mass at
\[
 t_n=\frac{1}{n}\arccos\left(1-\frac{n}{2}\right)
\]
works. The one-vertex graph is trivial. These complete-graph conclusions also appear in \cite[Theorem~8]{BCM}, and the times $t_n$ are the instantaneous mixing times of $K_2$, $K_3$, $K_4$ found in \cite{ABTW}.

A connected imprimitive strongly regular graph that is not complete is a complete multipartite graph with
equal part sizes. Write it as $K_{m,\ldots,m}$ with $\ell\ge2$ parts and $m\ge2$.
Its parameters include $n=\ell m$, $k=(\ell-1)m$, $r=0$, $s=-m$.
Then $Q=0$, so \eqref{eq:line} requires
\[
 y=1-\frac{n}{2}.
\]
As $y\ge-1$, necessarily $n\le4$. Thus among these graphs only
$K_{2,2}=C_4$ admits UAM, with the point mass at $\displaystyle\frac{\pi}{4}$. For complete bipartite graphs with more than seven vertices this is \cite[Proposition~9]{BCM}, which rests on the trace conditions of \cite[Lemma~7]{BCM}.
For clarity, the complete multipartite description follows because the complement of a connected imprimitive strongly regular graph is a disconnected strongly regular graph, hence a union of equal cliques. Vertices in different components force the nonadjacent common-neighbour parameter to be zero. If a component were not complete, a shortest path between two nonadjacent vertices would contain a pair at distance two, with a common neighbour. This is impossible. Hence the components are cliques, and regularity makes their sizes equal.

\section{Classification}\label{sec:classification}
In this section $G$ is a connected strongly regular graph on $n$ vertices that is not complete. It is \emph{primitive} if its complement is connected as well, equivalently $\nu\ge1$ and $n-2k+\lambda\ge1$. By Section~\ref{sec:edge}, $C_4$ is the only imprimitive graph of this kind with a UAM time law. If $G$ has parameters $(n,k,\lambda,\nu)$ and restricted eigenvalues $r>s$, its complement has parameters $(n,\,n-k-1,\,n-2k+\nu-2,\,n-2k+\lambda)$ and restricted eigenvalues $-1-s>-1-r$.

\subsection{The Gram condition and complex Hadamard matrices}

\begin{lemma}[Gram condition and complex Hadamard matrices]\label{lem:gram-hadamard}
The following two statements are equivalent.
\begin{enumerate}
 \item[(a)] The moment line \eqref{eq:line}, with $\eta$ allowed to be any real number, contains a point $(x,y,z)$ with $C(x,y,z)\succeq0$.
 \item[(b)] There are complex numbers $\zeta_1,\zeta_2$ of modulus one such that $\mathcal H=I+\zeta_1A+\zeta_2(J-I-A)$ satisfies $\mathcal H\mathcal H^*=nI$.
\end{enumerate}
If $G$ admits UAM, then (a) holds by Theorem~\ref{thm:line} and Remark~\ref{rem:gram}.
\end{lemma}
\begin{proof}
Index the rows and columns of $C(x,y,z)$ by $0,r,s$ in this order, so that $C_{0r}=x$, $C_{0s}=y$ and $C_{rs}=z$. For distinct vertices $g,h$ put $e_j=(E_j)_{gh}$ for $j\in\{0,r,s\}$. These numbers are real, and the left side of
\begin{equation}\label{eq:quadratic-entry}
 \sum_{j,l\in\{0,r,s\}}e_je_l\,C_{jl}(x,y,z)=\frac{1}{n}
\end{equation}
is the expression \eqref{eq:entry}. The proof of Theorem~\ref{thm:line} shows that \eqref{eq:quadratic-entry} holds for adjacent and for nonadjacent pairs if and only if $(x,y,z)$ lies on the moment line, as an identity in $x,y,z$.

Assume (a). The set of real $\eta$ with $C(\kappa+P\eta,\kappa+Q\eta,1-\eta)\succeq0$ is closed, convex, nonempty, and bounded because $|1-\eta|\le1$ on it. Let $\eta^*$ be its largest element and $C^*$ the corresponding matrix. If $C^*$ were positive definite, it would remain positive definite for slightly larger $\eta$. Hence $C^*$ is singular, so it has rank at most two and is the Gram matrix of three unit vectors in $\R^2$. There are therefore real numbers $\psi_0,\psi_r,\psi_s$ with $C^*_{jl}=\cos(\psi_j-\psi_l)$. Put $V=e^{i\psi_0}E_0+e^{i\psi_r}E_r+e^{i\psi_s}E_s$. Then
\[
 |V_{gh}|^2=\sum_{j,l}e_je_l\cos(\psi_j-\psi_l)=\sum_{j,l}e_je_lC^*_{jl}=\frac{1}{n}
\]
for all distinct $g,h$. Since $VV^*=E_0+E_r+E_s=I$, every row of the matrix $(|V_{gh}|^2)$ sums to one, so its diagonal entries equal $\displaystyle\frac{1}{n}$ as well. The matrix $V$ lies in the span of $I,A,J$, so its diagonal entries are all equal to one number $\zeta_0$ with $\displaystyle|\zeta_0|=\frac{1}{\sqrt n}$. Then $\mathcal H=\zeta_0^{-1}V$ has the form in (b), and $\mathcal H\mathcal H^*=|\zeta_0|^{-2}VV^*=nI$.

Conversely assume (b). The matrix $\mathcal H$ lies in the span of $E_0,E_r,E_s$, say $\mathcal H=\sum_jw_jE_j$, and $\mathcal H\mathcal H^*=\sum_j|w_j|^2E_j=nI$ gives $|w_j|=\sqrt n$. Put $\displaystyle u_j=\frac{w_j}{\sqrt n}$ and $C^*_{jl}=\operatorname{Re}(u_j\overline{u_l})$. Then $C^*$ is positive semidefinite with unit diagonal, and for distinct $g,h$
\[
 \sum_{j,l}e_je_lC^*_{jl}=\left|\sum_ju_je_j\right|^2=\frac{|\mathcal H_{gh}|^2}{n}=\frac{1}{n}.
\]
Hence the point $(C^*_{0r},C^*_{0s},C^*_{rs})$ lies on the moment line.
\end{proof}

\begin{theorem}[Chan {\cite[Theorem~2.5]{Chan11}}, see {\cite[Theorem~4.1]{GMR}}]\label{thm:chan}
Let $G$ be primitive and satisfy the equivalent conditions (a) and (b) of Lemma~\ref{lem:gram-hadamard}. Then $G$ or its complement has one of the following parameter sets, where $\theta$ denotes the positive restricted eigenvalue of the graph in the list.
\begin{enumerate}
 \item[(i)] $(4\theta^2,\,2\theta^2-\theta,\,\theta^2-\theta,\,\theta^2-\theta)$, the Latin square type.
 \item[(ii)] $(4\theta^2,\,2\theta^2+\theta,\,\theta^2+\theta,\,\theta^2+\theta)$, the negative Latin square type.
 \item[(iii)] $(4\theta^2-1,\,2\theta^2,\,\theta^2,\,\theta^2)$, the symplectic type.
 \item[(iv)] $(4\theta^2+4\theta+1,\,2\theta^2+2\theta,\,\theta^2+\theta-1,\,\theta^2+\theta)$, the conference graphs.
 \item[(v)] $(4\theta^2+4\theta+2,\,2\theta^2+\theta,\,\theta^2-1,\,\theta^2)$, the regular conference matrix type.
\end{enumerate}
\end{theorem}

In (iv) one has $\displaystyle\theta=\frac{\sqrt n-1}{2}$, which is an integer exactly when $n$ is a square. In the other four families and their complements $\displaystyle k\ne\frac{n-1}{2}$, so these are not conference parameters, and $\theta$ is a positive integer by Lemma~\ref{lem:conference}. Since $G$ is connected, primitive and not complete, $\theta\ge2$ in (i), (ii) and (iii), and $\theta\ge1$ in (v). Table~\ref{tab:families} lists the data needed below.

\begin{table}[ht]
\centering
\renewcommand{\arraystretch}{2.3}
\setlength{\tabcolsep}{4.5pt}
\begin{tabular}{@{}lllllll@{}}
\toprule
case & $n$ & $k$ & $r$ & $s$ & $P$ & $Q$\\
\midrule
(i) & $4\theta^2$ & $2\theta^2-\theta$ & $\theta$ & $-\theta$ & $\theta^2-1$ & $\theta^2$\\
(i) compl. & $4\theta^2$ & $2\theta^2+\theta-1$ & $\theta-1$ & $-\theta-1$ & $\theta^2$ & $\theta^2-1$\\
(ii) & $4\theta^2$ & $2\theta^2+\theta$ & $\theta$ & $-\theta$ & $\theta^2$ & $\theta^2-1$\\
(ii) compl. & $4\theta^2$ & $2\theta^2-\theta-1$ & $\theta-1$ & $-\theta-1$ & $\theta^2-1$ & $\theta^2$\\
(iii) & $4\theta^2-1$ & $2\theta^2$ & $\theta$ & $-\theta$ & $\dfrac{(2\theta-1)^2(\theta+1)}{4\theta}$ & $\dfrac{(2\theta+1)^2(\theta-1)}{4\theta}$\\
(iii) compl. & $4\theta^2-1$ & $2\theta^2-2$ & $\theta-1$ & $-\theta-1$ & $\dfrac{(2\theta+1)^2(\theta-1)}{4\theta}$ & $\dfrac{(2\theta-1)^2(\theta+1)}{4\theta}$\\
(iv) & $(2\theta+1)^2$ & $2\theta^2+2\theta$ & $\theta$ & $-\theta-1$ & $\theta^2+\theta$ & $\theta^2+\theta$\\
(v) & $4\theta^2+4\theta+2$ & $2\theta^2+\theta$ & $\theta$ & $-\theta-1$ & $\dfrac{(n-2)^2}{4(n-1)}$ & $\dfrac{n^2}{4(n-1)}$\\
(v) compl. & $4\theta^2+4\theta+2$ & $2\theta^2+3\theta+1$ & $\theta$ & $-\theta-1$ & $\dfrac{n^2}{4(n-1)}$ & $\dfrac{(n-2)^2}{4(n-1)}$\\
\bottomrule
\end{tabular}
\caption{The parameter families of Theorem~\ref{thm:chan} and their complements, with the constants of \eqref{eq:PQ}. The conference parameters (iv) are self-complementary.}\label{tab:families}
\end{table}

We do not use Theorem~\ref{thm:chan} in the sequel. The following elementary argument proves the part that we need, and together with Lemma~\ref{lem:gram-hadamard} it gives a short proof of Theorem~\ref{thm:chan}.

\begin{lemma}[Three inequalities]\label{lem:three-ineq}
Let $G$ be primitive, and put
\[
 \tau=2k-r-s-n,\qquad \Xi=n+2\tau-4\nu,\qquad \sigma=(r+s+1)^2 .
\]
Then $\tau$ and $\Xi$ are integers, $\sigma$ is the square of the integer $\lambda-\nu+1$, and
\begin{equation}\label{eq:j-sigma}
 \Delta^2=\tau^2+\Xi n,\qquad \sigma=\Delta^2-n+1-\Xi .
\end{equation}
If the moment line \eqref{eq:line} contains a point with $C(x,y,z)\succeq0$, then
\begin{equation}\label{eq:three-ineq}
 \Delta^2\le\frac{n^2}{n-4},\qquad \tau^2\le\frac{n^2}{n-4},\qquad \Delta^2\ge\frac{n(\sigma-1)}{2}.
\end{equation}
\end{lemma}
\begin{proof}
Recall $\omega_1=k-r$ and $\omega_2=k-s$, so that $\omega_2-\omega_1=\Delta$ and $\omega_1+\omega_2=n+\tau$. Applying \eqref{eq:srg} to the all-ones vector gives $k^2=(r+s)k-rs+n\nu$, that is $\omega_1\omega_2=n\nu$. Hence $4n\nu=(\omega_1+\omega_2)^2-(\omega_2-\omega_1)^2=(n+\tau)^2-\Delta^2$, which is the first identity in \eqref{eq:j-sigma}. Since $\nu=k+rs=k+(k-\omega_1)(k-\omega_2)$, the numbers $k$ and $k-r-s-1$ are the two roots of $X^2-(\omega_1+\omega_2-1)X+\omega_1\omega_2-\nu$. So the discriminant of this quadratic equals $\sigma$. Expanding the discriminant with $\omega_1\omega_2=n\nu$ gives $n\sigma=\tau^2+(n-1)(\Delta^2-n)$, and the first identity turns this into the second.

In \eqref{eq:PQ} the two factors of $P$ are positive and add up to $n$. The same holds for $Q$, because $n-k+s$ is the degree of the complement minus its positive restricted eigenvalue $-1-s$, and this difference is positive since the complement is connected. Hence $\displaystyle0<P,Q\le\frac{n^2}{4\Delta^2}$. A direct computation from \eqref{eq:PQ} and \eqref{eq:j-sigma} gives
\[
 P-Q=\frac{\tau}{\Delta},\qquad 2\kappa+2(P+Q)=\frac{n(1-\sigma)}{\Delta^2}.
\]
Let $(x,y,z)$ be a point of the moment line with $C(x,y,z)\succeq0$ and $\eta=1-z$. Then $|x|,|y|,|z|\le1$, so $0\le\eta\le2$, and $\eta>0$ because $x=\kappa<-1$ when $\eta=0$. From $x\ge-1$ we get $\displaystyle P\eta\ge\frac{n}{2}-2$ and hence
\[
 \eta\ge\frac{n-4}{2P}\ge\frac{2(n-4)\Delta^2}{n^2}.
\]
Together with $\eta\le2$ this is the first inequality in \eqref{eq:three-ineq}. Positive semidefiniteness of $C$ gives $(x-y)^2\le2(1-z)$, that is $\displaystyle\frac{\tau^2\eta^2}{\Delta^2}\le2\eta$. So $\displaystyle\eta\le\frac{2\Delta^2}{\tau^2}$ when $\tau\ne0$, and the lower bound for $\eta$ gives the second inequality. Finally $x+y=2\kappa+(P+Q)\eta$ increases with $\eta$, so
\[
 -2\le x+y\le2\kappa+2(P+Q)=\frac{n(1-\sigma)}{\Delta^2},
\]
which is the third inequality.
\end{proof}

\begin{proposition}[Elementary form of Theorem~\ref{thm:chan}]\label{prop:elementary-chan}
Let $G$ be primitive. If the moment line contains a point with $C(x,y,z)\succeq0$, then $\Xi,\sigma\in\{0,1\}$, and $G$ or its complement has one of the parameter sets (i) to (v) of Theorem~\ref{thm:chan}.
\end{proposition}
\begin{proof}
Let first $n\ge13$. By \eqref{eq:j-sigma} and the first two inequalities of \eqref{eq:three-ineq}, $\displaystyle|\Xi|\,n=|\Delta^2-\tau^2|\le\frac{n^2}{n-4}$, so $|\Xi|\le1$. The first and third inequalities give $\displaystyle\sigma\le1+\frac{2n}{n-4}<4$, and $\sigma$ is a square, so $\sigma\in\{0,1\}$. If $\Xi=-1$, then $\displaystyle\tau^2=\Delta^2+n\le\frac{n^2}{n-4}$ gives $\displaystyle\Delta^2\le\frac{4n}{n-4}<6$, while $\sigma\ge0$ gives $\Delta^2\ge n-2\ge11$. Hence $\Xi,\sigma\in\{0,1\}$.

The four cases determine the parameters, since $\displaystyle k=\frac{n+\tau+r+s}{2}$ and $r,s$ follow from $r+s$ and $\Delta$. In the last three cases below $\Delta$ is an integer, because nonintegral restricted eigenvalues occur only for the conference parameters by Lemma~\ref{lem:conference}.
\begin{itemize}
 \item $\Xi=1$ and $\sigma=0$. Then $r+s=-1$, $\tau=0$ and $\Delta^2=n$. These are the conference parameters (iv).
 \item $\Xi=0$ and $\sigma=0$. Then $r+s=-1$, $\tau=\pm\Delta$ and $\Delta^2=n-1$. With $\Delta=2\theta+1$ this is (v) for $\tau=-\Delta$ and its complement for $\tau=\Delta$.
 \item $\Xi=0$ and $\sigma=1$. Then $r+s\in\{0,-2\}$, $\tau=\pm\Delta$ and $\Delta^2=n$. With $\Delta=2\theta$ the four choices give (i), (ii) and their complements.
 \item $\Xi=1$ and $\sigma=1$. Then $r+s\in\{0,-2\}$, $\tau^2=1$ and $\Delta^2=n+1$, so $\Delta=2\theta$ and $n=4\theta^2-1$. For $r+s=0$ and $\tau=1$ this is (iii), and for $r+s=-2$ and $\tau=-1$ it is the complement of (iii). The two remaining sign choices give $k=2\theta^2-1$. Then the relation $k+fr+gs=0$ for the multiplicities $f,g$ of $r,s$, with $f+g=n-1$, forces $\theta$ to divide $2\theta^2-1$, which is impossible for $\theta\ge2$.
\end{itemize}
For $5\le n\le12$ we argue directly. If the restricted eigenvalues are not integers, Lemma~\ref{lem:conference} gives $n\equiv1\pmod4$ with $n$ not a square, so $n=5$ with parameters $(5,2,0,1)$. Otherwise $r\ge1$ and $s\le-2$. Indeed $rs=\nu-k<0$, as $\nu=k$ would make $G$ complete multipartite, and likewise $(r+1)(s+1)<0$ for the complement. Put $d=r|s|$ and $c=(r+1)(|s|-1)$. Then $k=\nu+d$ and $k-\lambda-1=c$, so the relation $k(k-\lambda-1)=(n-k-1)\nu$ gives
\[
 n-1=\nu+c+d+\frac{cd}{\nu}\ge\left(\sqrt c+\sqrt d\right)^2 .
\]
For $n\le12$ this leaves only $(r,s)=(1,-2)$, where $c=d=2$ and $\nu$ divides $4$. The values $\nu=1,2,4$ give $(10,3,0,1)$, $(9,4,1,2)$ and $(10,6,3,4)$. These four sets lie in (iv), (v), (iv) and the complement of (v), with $(\Xi,\sigma)=(1,0)$ for $n=5,9$ and $(0,0)$ for $n=10$.
\end{proof}

Conversely, every parameter set in (i) to (v) and every complementary set satisfies the conditions of Lemma~\ref{lem:gram-hadamard}. For (iv) the point $(0,0,z_*)$ of \eqref{eq:conference-target} lies on the moment line \eqref{eq:conference-line}, and $C$ is positive definite there. For (i), (ii) and their complements, the proof of Theorem~\ref{thm:gmr-corrected} shows that $J-2A$ or $-(J-2A-2I)$ is a Hadamard matrix of the form in Lemma~\ref{lem:gram-hadamard}(b). For (iii), (v) and their complements, Propositions~\ref{prop:family-v} and~\ref{prop:family-iii} give UAM laws. So Lemma~\ref{lem:gram-hadamard} and Proposition~\ref{prop:elementary-chan} prove Theorem~\ref{thm:chan} together with its converse, which is the content of \cite[Theorem~2.5]{Chan11}.

\begin{corollary}\label{cor:chan-families}
If a primitive strongly regular graph admits UAM, then it or its complement has one of the parameter sets (i) to (v) of Theorem~\ref{thm:chan}.
\end{corollary}
\begin{proof}
A UAM law gives a point of the moment line with $C(x,y,z)\succeq0$ by Theorem~\ref{thm:line} and Remark~\ref{rem:gram}, so Proposition~\ref{prop:elementary-chan} applies.
\end{proof}

\subsection{The five families}

\begin{proposition}[Latin square types]\label{prop:latin}
Let $G$ or its complement have parameters (i) or (ii). The following statements are equivalent.
\begin{enumerate}
 \item[(a)] $G$ admits UAM.
 \item[(b)] $G$ admits instantaneous uniform mixing.
 \item[(c)] $\theta$ is even.
\end{enumerate}
In this case a probability law is a UAM time law if and only if it is concentrated on the set $\displaystyle S_\theta=\left\{\frac{(2j+1)\pi}{2\theta}\,\middle|\,j=0,1,2,\ldots\right\}$, and every point of $S_\theta$ is a time of instantaneous uniform mixing. In particular no absolutely continuous law gives UAM.
\end{proposition}
\begin{proof}
By Table~\ref{tab:families}, in all four cases $\Delta=2\theta$, $\kappa=1-2\theta^2$ and $\{P,Q\}=\{\theta^2-1,\theta^2\}$. Moreover the frequencies $k-r$ and $k-s$ are $2\theta^2$ and one of $2\theta(\theta-1)$, $2\theta(\theta+1)$, and in each case the frequency $2\theta^2$ carries the coefficient $\theta^2$ while the other frequency carries $\theta^2-1$. The moment with coefficient $\theta^2-1$ is at least $-1$ only if $(\theta^2-1)\eta\ge2\theta^2-2$, that is $\eta\ge2$. Hence every UAM law has $\eta=2$ and $z=\displaystyle\int\cos(2\theta t)\,d\mu=-1$, so it is concentrated on $S_\theta$. At $\eta=2$ the line requires the moment with coefficient $\theta^2$ to equal $1$ and the other moment to equal $-1$. For $t\in S_\theta$ and an integer $m$ one has $\cos(2\theta mt)=(-1)^m$. So the moment of frequency $2\theta^2$ equals $(-1)^\theta$, and the moment of frequency $2\theta(\theta\pm1)$ equals $(-1)^{\theta\pm1}$. Both requirements hold if and only if $\theta$ is even. In that case every $t\in S_\theta$ gives the point of the moment line with $\eta=2$, so $\delta_t$ is a UAM law, which means instantaneous uniform mixing at $t$, and every law concentrated on $S_\theta$ is a UAM law. If $\theta$ is odd, no law satisfies the requirements.
\end{proof}

\begin{proposition}[Square conference parameters]\label{prop:square-conference}
Let $G$ have parameters (iv) with $\theta$ a positive integer, so that $n=(2\theta+1)^2$. Then $G$ admits UAM if and only if $\theta=1$.
\end{proposition}
\begin{proof}
Here $\delta=2\theta+1$, $a_0=\theta$, $b_0=1$ and $D=\theta+1$. Let $\varrho$ be the image of a time law $\mu$ under $t\mapsto\xi=\cos((2\theta+1)t)$, and put $w=1+\xi\in[0,2]$ and $\displaystyle\overline w=\int w\,d\varrho$. Then $\displaystyle x=\int\mathsf T_\theta\,d\varrho$, $\displaystyle y=\int\mathsf T_{\theta+1}\,d\varrho$ and $\displaystyle z=\int\xi\,d\varrho=\overline w-1$. Since $P=Q$, a UAM law satisfies $x=y$, and \eqref{eq:conference-line} becomes
\begin{equation}\label{eq:square-conf}
 x=y=\frac{1}{2}-(\theta^2+\theta)\,\overline w.
\end{equation}
From $\mathsf T_m'(\cos\phi)\sin\phi=m\sin(m\phi)$ and $|\sin(m\phi)|\le m|\sin\phi|$ we get Markov's bound $|\mathsf T_m'|\le m^2$ on $[-1,1]$. Hence $\mathsf T_m(1-w)\ge1-m^2w$ for $w\in[0,2]$, while $\mathsf T_m(-1+w)=(-1)^m\mathsf T_m(1-w)$.

Let $m_e$ be the even one of $\theta$ and $\theta+1$. The corresponding moment is at least $1-m_e^2\overline w$, so \eqref{eq:square-conf} gives $\displaystyle(m_e^2-\theta^2-\theta)\,\overline w\ge\frac{1}{2}$. If $\theta$ is even, then $m_e=\theta$ and the left side equals $-\theta\overline w\le0$, a contradiction. If $\theta$ is odd, then $m_e=\theta+1$ and $\displaystyle\overline w\ge\frac{1}{2(\theta+1)}$. The other moment is at least $-1$, so \eqref{eq:square-conf} gives $\displaystyle\overline w\le\frac{3}{2\theta(\theta+1)}$. Together these force $\theta\le3$. If $\theta=3$, equality holds throughout, so $\displaystyle\overline w=\frac{1}{8}$ and $x=y=-1$. Then $\mathsf T_3(\xi)=-1$ and $\mathsf T_4(\xi)=-1$ hold $\varrho$-almost everywhere. This is impossible, because $\mathsf T_3(\xi)=-1$ only for $\displaystyle\xi\in\left\{-1,\frac{1}{2}\right\}$, where $\mathsf T_4$ takes the values $1$ and $\displaystyle-\frac{1}{2}$. Hence $\theta=1$, and for $n=9$ the law \eqref{eq:nine} gives UAM.
\end{proof}

\begin{proposition}[Regular conference matrix type]\label{prop:family-v}
Let $G$ or its complement have parameters (v), so $n=4\theta^2+4\theta+2$ with $\theta\ge1$. Put
\[
 t_*=\frac{4\pi\theta^2}{n},\qquad
 p=\frac{2(n-1)}{n\left(1+\cos\dfrac{2\pi}{n}\right)}.
\]
Then $\mu=(1-p)\delta_0+p\,\delta_{t_*}$ is a UAM time law for $G$. The same law serves for $G$ and for its complement.
\end{proposition}
\begin{proof}
By Table~\ref{tab:families}, $\Delta=2\theta+1$, so $\Delta^2=n-1$. For (v) one has $\displaystyle k-s=\frac{n}{2}$, and for its complement $\displaystyle k-r=\frac{n}{2}$. Hence $\displaystyle\frac{n}{2}\,t_*=2\pi\theta^2\in2\pi\Z$. For (v) this gives $y=1$ at both atoms, and $\cos((k-r)t)=\cos((k-s)t-\Delta t)=\cos(\Delta t)$ gives $x=z$. For the complement it gives $x=1$ and $y=z$. Next $4\theta(\theta+1)=n-2$ and $\displaystyle\theta^2(2\theta+1)+\theta(\theta+1)=\frac{\theta n}{2}$ give
\[
 \Delta t_*=2\pi\theta-\frac{\pi(n-2)}{n},\qquad \cos(\Delta t_*)=-\cos\frac{2\pi}{n}.
\]
Therefore $\mu$ has $\displaystyle z=1-p\left(1+\cos\frac{2\pi}{n}\right)=-\frac{n-2}{n}$, and $\displaystyle\eta=1-z=\frac{2(n-1)}{n}$. For (v) the line requires
\[
 \kappa+Q\eta=1-\frac{n}{2}+\frac{n^2}{4(n-1)}\cdot\frac{2(n-1)}{n}=1,\qquad
 \kappa+P\eta=1-\frac{n}{2}+\frac{(n-2)^2}{2n}=-\frac{n-2}{n},
\]
which are the values $y=1$ and $x=z$ found above. For the complement $P$ and $Q$ are interchanged, and so are the roles of $x$ and $y$. Finally $0<p<1$, because $\displaystyle\cos\frac{2\pi}{n}>1-\frac{2\pi^2}{n^2}\ge1-\frac{2}{n}$ for $n\ge\pi^2$, and $n\ge10$.
\end{proof}

For $\theta=1$ this is the Petersen law \eqref{eq:petersen}, since $\displaystyle t_*=\frac{2\pi}{5}$ and $\displaystyle p=\frac{9(5-\sqrt5)}{25}$.

\begin{proposition}[Symplectic type]\label{prop:family-iii}
Let $n=4\theta^2-1$ with $\theta\ge2$.
\begin{enumerate}
 \item[(a)] Let $G$ have parameters (iii). Let $e$ be the unique solution in $\displaystyle\left(0,\frac{\pi}{4\theta}\right)$ of $\displaystyle\sin(2\theta e)=\frac{\cos e}{2\theta}$, and put
 \[
  t_2=\frac{1}{\theta}\left(\frac{\pi}{2}+(-1)^\theta e\right),\qquad
  g=\sin((2\theta-1)e)-\frac{1}{2\theta}+2P\sin^2e,\qquad
  p=\frac{-2g}{n-2g}.
 \]
 Then $g<0$, and $\mu=p\,\delta_0+(1-p)\delta_{t_2}$ is a UAM time law.
 \item[(b)] Let $G$ have the complementary parameters $(4\theta^2-1,\,2\theta^2-2,\,\theta^2-3,\,\theta^2-1)$. Put
 \[
  t_0=\frac{\pi}{2\theta},\quad t_1=\frac{3\pi}{2\theta},\quad
  \chi_j=\cos((k-r)t_j),\quad
  p_0=\frac{2\theta\chi_1+1}{2\theta(\chi_1-\chi_0)}.
 \]
 Then $0<p_0<1$, and $\mu=p_0\,\delta_{t_0}+(1-p_0)\delta_{t_1}$ is a UAM time law.
\end{enumerate}
\end{proposition}
\begin{proof}
We use the function $R=(R_1,R_2)$ of \eqref{eq:R}. A law $\mu$ is a UAM time law if and only if $\displaystyle\int R\,d\mu=0$.

(a) Put $\phi=\theta t$. By Table~\ref{tab:families}, $\omega_1=\theta(2\theta-1)$, $\omega_2=\theta(2\theta+1)$, $\omega_3=2\theta$ and $\displaystyle P-Q=\frac{1}{2\theta}$, so
\[
 R_1-R_2=\cos((2\theta-1)\phi)-\cos((2\theta+1)\phi)-\frac{1-\cos(2\phi)}{2\theta}
 =2\sin\phi\left(\sin(2\theta\phi)-\frac{\sin\phi}{2\theta}\right).
\]
The function $\displaystyle\sin(2\theta e)-\frac{\cos e}{2\theta}$ increases on $\displaystyle\left[0,\frac{\pi}{4\theta}\right]$, is negative at $0$ and positive at $\displaystyle\frac{\pi}{4\theta}$, so $e$ exists and is unique. At $\displaystyle\phi_2=\theta t_2=\frac{\pi}{2}+(-1)^\theta e$ one has $\sin\phi_2=\cos e$ and $\sin(2\theta\phi_2)=\sin(2\theta e)$, hence $R_1(t_2)=R_2(t_2)$. Using $\cos(2\phi_2)=-\cos(2e)$, $\cos((2\theta-1)\phi_2)=\sin((2\theta-1)e)$ and $\displaystyle-\kappa-2P=-\frac{1}{2\theta}$, the common value is $R_1(t_2)=g$.

To see $g<0$, note that $\displaystyle2\theta e\le\arcsin\frac{1}{2\theta}$. For $\displaystyle0\le s\le\frac{1}{4}$ one has $\displaystyle\frac{1}{\sqrt{1-s^2}}-1\le\frac{3s^2}{5}$, hence $\displaystyle\arcsin y\le y+\frac{y^3}{5}$ for $\displaystyle0\le y\le\frac{1}{4}$, and
\[
 e\le e_\theta:=\frac{1}{4\theta^2}+\frac{1}{80\theta^4}.
\]
With $\sin u\le u$ for $u\ge0$ and $\displaystyle2P=2\theta^2-\frac{3}{2}+\frac{1}{2\theta}\le2\theta^2$, the number $g$ is at most the value of the increasing function $\displaystyle u\mapsto(2\theta-1)u-\frac{1}{2\theta}+2\theta^2u^2$ at $u=e_\theta$, which is
\[
 -\frac{1}{8\theta^2}+\frac{1}{40\theta^3}+\frac{1}{3200\theta^6}<0.
\]
Finally $\displaystyle R(0)=\frac{n}{2}(1,1)$ and $R(t_2)=g(1,1)$, so $pR(0)+(1-p)R(t_2)=0$ with $0<p<1$.

(b) Here $\Delta=2\theta$, so $\cos(\Delta t_j)=-1$ and $\eta=2$. By Table~\ref{tab:families} the line at $\eta=2$ requires
\[
 x=\kappa+2P=-\frac{1}{2\theta},\qquad y=\kappa+2Q=\frac{1}{2\theta}.
\]
Since $k-s=(k-r)+2\theta$, every $t$ with $\cos(2\theta t)=-1$ has $\cos((k-s)t)=-\cos((k-r)t)$, so $y=-x$ for every law concentrated on such times. From $\displaystyle\frac{k-r}{2\theta}=\frac{(2\theta+1)(\theta-1)}{2\theta}=\theta-\frac{1}{2}-\frac{1}{2\theta}$ one obtains
\[
 \chi_j=(-1)^{\theta+j+1}\sin\frac{(2j+1)\pi}{2\theta}.
\]
Thus $\chi_0$ and $\chi_1$ have opposite signs. The negative one has absolute value at least $\displaystyle\frac{1}{\theta}$ by $\displaystyle\sin u\ge\frac{2u}{\pi}$ on $\displaystyle\left[0,\frac{\pi}{2}\right]$, since it is $\chi_0$ when $\theta$ is even and $\chi_1$ with $\displaystyle\frac{3\pi}{2\theta}\le\frac{\pi}{2}$ when $\theta\ge3$ is odd. Hence $0<p_0<1$ and $\displaystyle p_0\chi_0+(1-p_0)\chi_1=-\frac{1}{2\theta}$, which gives the required moments.
\end{proof}

The laws of Propositions~\ref{prop:family-v} and~\ref{prop:family-iii} use two observation times, and one time does not suffice. For this we use the following $p$-adic form of Niven's theorem. For a rational number $q$ and a prime $p$ let $\operatorname{ord}_p^-(q)=\max\{0,-v_p(2q)\}$, where $v_p$ is the $p$-adic valuation.

\begin{lemma}[Rational cosines]\label{lem:rational-cos}
Let $\omega,\omega'$ be positive integers and $t$ a real number such that $q=\cos(\omega t)$ and $q'=\cos(\omega't)$ are rational. Then $\omega'\operatorname{ord}_p^-(q)=\omega\operatorname{ord}_p^-(q')$ for every prime $p$.
\end{lemma}
\begin{proof}
Put $a=e^{i\omega t}$ and $b=e^{i\omega't}$, so that $a^{\omega'}=b^{\omega}$. The number $a$ is a root of $X^2-2qX+1$, whose other root is $a^{-1}$. Fix a prime ideal $\mathfrak P$ above $p$ in a number field containing $a$ and $b$, and let $v_{\mathfrak P}$ be its valuation, so that $v_{\mathfrak P}(p)=e_{\mathfrak P}$ is the ramification index. If $v_p(2q)\ge0$, then $a$ and $a^{-1}$ are integral at $\mathfrak P$, so $v_{\mathfrak P}(a)=0$. If $v_p(2q)<0$, then $v_{\mathfrak P}(a)=-v_{\mathfrak P}(a^{-1})$ and $v_{\mathfrak P}(a+a^{-1})=e_{\mathfrak P}v_p(2q)<0$, so $|v_{\mathfrak P}(a)|=-e_{\mathfrak P}v_p(2q)$. In both cases $|v_{\mathfrak P}(a)|=e_{\mathfrak P}\operatorname{ord}_p^-(q)$, and likewise $|v_{\mathfrak P}(b)|=e_{\mathfrak P}\operatorname{ord}_p^-(q')$. Applying $v_{\mathfrak P}$ to $a^{\omega'}=b^{\omega}$ proves the lemma.
\end{proof}

\begin{proposition}[No single observation time]\label{prop:no-single-time}
No strongly regular graph whose parameters or complementary parameters lie in (iii) with $\theta\ge2$ or in (v) admits instantaneous uniform mixing.
\end{proposition}
\begin{proof}
For a point mass $\delta_t$ the matrix $C(x,y,z)$ is the Gram matrix of three unit vectors in the plane, so $\det C=0$, and $x$, $y$, $z$ are the cosines of $\omega_1t$, $\omega_2t$, $\omega_3t$. Write $u=1+z$. By Table~\ref{tab:families} the restriction of $\det C$ to the moment line is
\[
 \det C=\frac{n(u-2)(nu-2)\left((n-2)^2u-2n\right)}{8(n-1)^2}
\]
for (v) and its complement, and
\[
 \det C=\frac{n\,u(u-2)(\Theta u-2)}{8\theta^2},\qquad \Theta=(4\theta^2-1)(\theta^2-1),
\]
for (iii) and its complement. The zero $u=2$ means $z=1$ and $x=\kappa<-1$, which is impossible. At every other zero the coordinates are rational, and Lemma~\ref{lem:rational-cos} applied to two of them with a suitable prime gives a contradiction, because the frequencies are positive.
\begin{itemize}
 \item (v), $\displaystyle u=\frac{2}{n}$. Here $y=1$ and $\displaystyle z=-\frac{n-2}{n}$. For an odd prime $p$ dividing $n$ one has $\operatorname{ord}_p^-(y)=0$ and $\operatorname{ord}_p^-(z)=v_p(n)>0$.
 \item (v), $\displaystyle u=\frac{2n}{(n-2)^2}$. Here $x=-1$ and $\displaystyle z=-1+\frac{2n}{(n-2)^2}$. Since $n-2=4\theta(\theta+1)$ and $v_2(4n)=3$, one has $\operatorname{ord}_2^-(x)=0$ and $\operatorname{ord}_2^-(z)=2v_2(n-2)-3>0$.
 \item (iii), $u=0$. Here $z=-1$ and $\displaystyle x=\frac{1}{2\theta}$. For a prime $p$ dividing $\theta$ one has $\operatorname{ord}_p^-(x)=v_p(\theta)>0=\operatorname{ord}_p^-(z)$.
 \item (iii), $\displaystyle u=\frac{2}{\Theta}$. Here $\displaystyle x=\frac{2\theta-3}{2(2\theta+1)(\theta-1)}$ and $\displaystyle z=\frac{2}{\Theta}-1$. For a prime $p$ dividing $2\theta-1$ one has $\operatorname{ord}_p^-(x)=0$ and $\operatorname{ord}_p^-(z)=v_p(\Theta)>0$.
\end{itemize}
For the complements the roles of $x$ and $y$ are interchanged. In the last two items this means $\displaystyle x=-\frac{1}{2\theta}$ and $\displaystyle x=-\frac{2\theta+3}{2(2\theta-1)(\theta+1)}$, and then a prime $p$ dividing $\theta$, respectively $2\theta+1$, works in the same way.
\end{proof}

\subsection{The classification}

\begin{theorem}[Strongly regular graphs with UAM]\label{thm:classification}
Let $G$ be a connected strongly regular graph on $n$ vertices that is not complete. Then $G$ admits a UAM time law if and only if one of the following holds.
\begin{enumerate}
 \item[(a)] $G$ is a conference graph, and $n$ is not a square or $n=9$.
 \item[(b)] $G$ or its complement has parameters $(4\theta^2+4\theta+2,\,2\theta^2+\theta,\,\theta^2-1,\,\theta^2)$ with $\theta\ge1$.
 \item[(c)] $G$ or its complement has parameters $(4\theta^2-1,\,2\theta^2,\,\theta^2,\,\theta^2)$ with $\theta\ge2$.
 \item[(d)] $G$ or its complement has parameters $(4\theta^2,\,2\theta^2-\theta,\,\theta^2-\theta,\,\theta^2-\theta)$ or $(4\theta^2,\,2\theta^2+\theta,\,\theta^2+\theta,\,\theta^2+\theta)$ with $\theta\ge2$ even.
 \item[(e)] $G$ is the $4$-cycle $C_4$.
\end{enumerate}
\end{theorem}
\begin{proof}
An imprimitive $G$ admits UAM exactly when $G=C_4$ by Section~\ref{sec:edge}. A primitive $G$ with UAM belongs, up to complement, to one of the families (i) to (v) by Corollary~\ref{cor:chan-families}. Families (i) and (ii) admit UAM exactly for even $\theta$ by Proposition~\ref{prop:latin}, which is (d). Family (iii) and its complement admit UAM by Proposition~\ref{prop:family-iii}, which is (c). Family (v) and its complement admit UAM by Proposition~\ref{prop:family-v}, which is (b). In family (iv), a nonsquare order gives UAM by Theorem~\ref{thm:conference}, and a square order gives UAM exactly when $n=9$ by Proposition~\ref{prop:square-conference}, which is (a).
\end{proof}

\begin{corollary}[Number of observation times]\label{cor:times}
For the graphs of Theorem~\ref{thm:classification}, one observation time suffices exactly in case (d), for $n=9$ and for $C_4$. In these cases every UAM law is concentrated on a discrete set. In all other cases exactly two observation times are necessary and sufficient.
\end{corollary}
\begin{proof}
Case (d) is Proposition~\ref{prop:latin}, the order nine is Section~\ref{sec:examples}, and for $C_4$ the moment line forces $y=\displaystyle\int\cos(4t)\,d\mu=-1$. Two times always suffice by Theorem~\ref{thm:twoatoms}. A single time is impossible for conference graphs of nonsquare order by Lemma~\ref{lem:no-single-time}, and for the families (iii) and (v) and their complements by Proposition~\ref{prop:no-single-time}.
\end{proof}

\subsection{Bounded densities}

Write $R'$ for the derivative of the function $R$ in \eqref{eq:R}.

\begin{lemma}[Spreading two atoms]\label{lem:spread}
Let $\mu=p\,\delta_{t_1}+(1-p)\,\delta_{t_2}$ be a UAM time law with $0<p<1$ and $R(t_1)\ne0$. If $R(t_1)$ and $R'(t_1)$ are linearly independent, then $G$ admits a UAM time law with a bounded density of compact support.
\end{lemma}
\begin{proof}
Since $pR(t_1)+(1-p)R(t_2)=0$, the vector $R(t_2)$ is a negative multiple of $R(t_1)$. Let $L$ be the line spanned by $R(t_1)$. As $R'(0)=0$, we have $t_1>0$. By the linear independence, all times $t'<t_1<t''$ sufficiently close to $t_1$ have $t'>0$, and $R(t')$, $R(t'')$ lie strictly on opposite sides of $L$. The segment joining $R(t')$ and $R(t'')$ meets $L$ at a point close to $R(t_1)$, and this point and $R(t_2)$ lie on opposite sides of $0$. Hence $0$ is an interior point of the triangle with vertices $R(t')$, $R(t'')$ and $R(t_2)$.

Interior points of a triangle stay interior under small perturbations of its vertices, and $R$ is continuous. So there are compact intervals $I_1,I_2,I_3\subset[0,\infty)$ of positive length, containing $t'$, $t''$ and $t_2$ respectively, such that $0$ is an interior point of the triangle with vertices $\bar R_1,\bar R_2,\bar R_3$ whenever each $\bar R_j$ lies in the closed convex hull of $R(I_j)$. Let $\mu_j$ be the uniform probability law on $I_j$ and $\displaystyle\bar R_j=\int R\,d\mu_j$. Then $0=p_1\bar R_1+p_2\bar R_2+p_3\bar R_3$ with positive weights of sum one, and $p_1\mu_1+p_2\mu_2+p_3\mu_3$ is a UAM time law with a bounded density supported in $I_1\cup I_2\cup I_3$.
\end{proof}

\begin{proposition}[Bounded densities]\label{prop:density}
A connected strongly regular graph that is not complete and admits UAM has a UAM time law with a bounded density of compact support if and only if it does not admit instantaneous uniform mixing.
\end{proposition}
\begin{proof}
If $G$ admits instantaneous uniform mixing, then Corollary~\ref{cor:times} and its proof show that $G$ is in case (d) of Theorem~\ref{thm:classification}, has order nine, or is $C_4$, and that every UAM law is concentrated on a discrete set. Such a law has no density.

Conversely, let $G$ admit UAM but not instantaneous uniform mixing. By Theorem~\ref{thm:classification} and Corollary~\ref{cor:times}, $G$ is a conference graph of nonsquare order, and then Theorem~\ref{thm:conference} applies, or $G$ or its complement has parameters (iii) or (v). In the latter case we apply Lemma~\ref{lem:spread} to the laws of Propositions~\ref{prop:family-v} and~\ref{prop:family-iii}, whose atoms satisfy $R\ne0$ because there is no instantaneous uniform mixing. It remains to check the linear independence of $R$ and $R'$ at one atom. By \eqref{eq:R} and $\displaystyle P-Q=\frac{\tau}{\Delta}$ from the proof of Lemma~\ref{lem:three-ineq},
\begin{align*}
 (R_1-R_2)'(t)&=-\omega_1\sin(\omega_1t)+\omega_2\sin(\omega_2t)-\tau\sin(\Delta t),\\
 (R_1+R_2)'(t)&=-\omega_1\sin(\omega_1t)-\omega_2\sin(\omega_2t)-(P+Q)\Delta\sin(\Delta t).
\end{align*}
\begin{itemize}
 \item (v) and its complement, at the atom $t_*$. Here $R(t_*)$ is a multiple of $\displaystyle R(0)=\frac{n}{2}(1,1)$, so it suffices that $(R_1-R_2)'(t_*)\ne0$. For (v) one has $\tau=-\Delta$, $\omega_2t_*\in2\pi\Z$ and $\displaystyle\sin(\omega_1t_*)=-\sin(\Delta t_*)=\sin\frac{2\pi}{n}$. For the complement one has $\tau=\Delta$, $\omega_1t_*\in2\pi\Z$ and $\displaystyle\sin(\omega_2t_*)=\sin(\Delta t_*)=-\sin\frac{2\pi}{n}$. In both cases $\displaystyle(R_1-R_2)'(t_*)=-\frac{n}{2}\sin\frac{2\pi}{n}\ne0$.
 \item (iii), at the atom $t_2$. Here $R(t_2)=g(1,1)$. With $\phi=\theta t$ the bracket in $\displaystyle R_1-R_2=2\sin\phi\left(\sin(2\theta\phi)-\frac{\sin\phi}{2\theta}\right)$ vanishes at $\phi_2=\theta t_2$, so
 \[
  (R_1-R_2)'(t_2)=2\theta\sin\phi_2\left(2\theta\cos(2\theta\phi_2)-\frac{\cos\phi_2}{2\theta}\right)
  =(-1)^\theta\,2\theta\cos e\left(2\theta\cos(2\theta e)+\frac{\sin e}{2\theta}\right),
 \]
 which is nonzero because $\displaystyle0<2\theta e<\frac{\pi}{2}$.
 \item The complement of (iii), at the atom $t_0$. Here $\displaystyle R(t_0)=\left(\chi_0+\frac{1}{2\theta}\right)(1,-1)$, so it suffices that $(R_1+R_2)'(t_0)\ne0$. Since $\Delta t_0=\pi$, one has $\sin(\Delta t_0)=0$ and $\sin(\omega_2t_0)=-\sin(\omega_1t_0)$, hence $(R_1+R_2)'(t_0)=2\theta\sin(\omega_1t_0)$. This is nonzero, because $\displaystyle\frac{\omega_1t_0}{\pi}=\frac{(2\theta+1)(\theta-1)}{2\theta}$ is not an integer, as $(2\theta+1)(\theta-1)\equiv-1\pmod\theta$.
\end{itemize}
\end{proof}

\section{A correction to the classification of instantaneous uniform mixing}\label{sec:gmr}
As printed, \cite[Theorem~5.1]{GMR} states that a primitive strongly regular graph with adjacency matrix $A$ has instantaneous uniform mixing exactly when (a) $J-2A$ is a regular symmetric Hadamard matrix of order $4\theta^2$ with constant diagonal and positive row sum and $\theta$ is even, or (b) $J-2A-2I$ is such a matrix and $\theta$ is odd, or (c) the graph is the Paley graph of order nine. In \cite{GMR} the letter $\theta$ also denotes the positive restricted eigenvalue of the graph, for instance in \cite[Lemma~5.2]{GMR}, while the order $4\theta^2$ forces $\displaystyle\theta=\frac{\sqrt n}{2}$. For the graphs in (b) these two readings differ by one. The statement has two defects.
\begin{itemize}
 \item Under either reading, the condition of positive row sum excludes every graph with parameters (ii) and every complement of a graph with parameters (i). For these graphs the Hadamard matrix among $J-2A$ and $J-2A-2I$ has row sum $-2\theta$, and by Proposition~\ref{prop:latin} they admit instantaneous uniform mixing when $\theta$ is even.
 \item Under the literal reading $\displaystyle\theta=\frac{\sqrt n}{2}$ the parity in (b) has to be even. The parity also fails in \cite[Lemma~5.2(b)]{GMR} for the graphs with parameters (ii), where both readings agree. Under the eigenvalue reading, the hypothesis of \cite[Lemma~5.2]{GMR} does not cover the intended complements of (i) and (ii). For example, the graph with parameters $(16,9,4,6)$ has positive restricted eigenvalue $1$, whereas its order is $4\cdot2^2$.
\end{itemize}
The following statement avoids $\theta$ altogether. For a regular graph the matrices $J-2A$ and $J-2A-2I$ are automatically symmetric and regular with constant diagonal.

\begin{theorem}\label{thm:gmr-corrected}
A primitive strongly regular graph $X$ on $n$ vertices with adjacency matrix $A$ admits instantaneous uniform mixing if and only if $X$ is the Paley graph of order nine, or $n$ is divisible by $16$ and one of the matrices $J-2A$ and $J-2A-2I$ is a Hadamard matrix. In the second case the times of instantaneous uniform mixing are exactly the odd multiples of $\displaystyle\frac{\pi}{\sqrt n}$.
\end{theorem}
\begin{proof}
The symmetric matrix $J-2A$ has entries $\pm1$, and it is a Hadamard matrix exactly when its eigenvalues $n-2k$, $-2r$, $-2s$ all have modulus $\sqrt n$. This means $n=4\theta^2$, $r=\theta$, $s=-\theta$ and $n-2k=\pm2\theta$, which are the parameter sets (i) and (ii). Similarly $J-2A-2I$ is a Hadamard matrix exactly when $n-2k-2$, $-2r-2$, $-2s-2$ have modulus $2\theta$ with $n=4\theta^2$, which gives $r=\theta-1$, $s=-\theta-1$ and the complements of (i) and (ii). This correspondence is classical, see \cite[Section~10.5]{BH}. In all these cases $16\mid n$ means that $\theta$ is even, and then Proposition~\ref{prop:latin} shows that the times of instantaneous uniform mixing are exactly the elements of $S_\theta$, the odd multiples of $\displaystyle\frac{\pi}{2\theta}=\frac{\pi}{\sqrt n}$. If $\theta$ is odd there is no such time. The Paley graph of order nine mixes uniformly by \eqref{eq:nine}.

Conversely, a graph with instantaneous uniform mixing admits UAM, so it is one of the graphs in Theorem~\ref{thm:classification}(a) to (d). Conference graphs of nonsquare order are excluded by Lemma~\ref{lem:no-single-time}, and among square orders only $n=9$ remains by Proposition~\ref{prop:square-conference}. The families (iii) and (v) and their complements are excluded by Proposition~\ref{prop:no-single-time}. The remaining graphs are those of Theorem~\ref{thm:classification}(d), which are the Hadamard cases with $16\mid n$.
\end{proof}

\begin{example}[Counterexamples]\label{ex:clebsch}
The halved $5$-cube has parameters $(16,10,6,6)$, which is (ii) with $\theta=2$. Its matrix $J-2A$ is a Hadamard matrix with row sum $-4$, while $J-2A-2I$ has the eigenvalue $-6$. So it satisfies none of the conditions of \cite[Theorem~5.1]{GMR}, under either reading of $\theta$, and \cite[Lemma~5.2(b)]{GMR} excludes it as well. Nevertheless it mixes uniformly at time $\displaystyle\frac{\pi}{4}$ by Theorem~\ref{thm:gmr-corrected}, in agreement with \cite[Theorem~5.2]{Chan20}. The complement of $K_4\square K_4$, with parameters $(16,9,4,6)$, is a second example of this kind. By \cite[Lemma~2.1]{GMR} it mixes uniformly at time $\displaystyle\frac{\pi}{4}$ together with $K_4\square K_4$, which falls under (a).

The Clebsch graph is the complement of the halved $5$-cube, with parameters $(16,5,0,2)$. Its matrix $J-2A-2I$ is a Hadamard matrix with diagonal $-1$ and row sum $4$, so under the literal reading it falls under (b) with $\theta=2$ and is excluded. However
\[
 U\left(\frac{\pi}{4}\right)=\frac{1}{2}\exp\left(\frac{i\pi}{4}\right)\left(A+I-\frac{1}{2}J\right)
\]
has all entries of modulus $\displaystyle\frac{1}{4}$, see also \cite[Theorem~1]{BKMT} and \cite[Theorem~6.1]{Chan20}. In the other direction, strongly regular graphs with parameters $(36,14,4,6)$ exist \cite[Section~9.9]{BH}. Under the literal reading they fall under (b) with $\theta=3$, but by Proposition~\ref{prop:latin} no graph with these parameters admits UAM. Their complements, with parameters $(36,21,12,12)$, are included by \cite[Lemma~5.2(b)]{GMR} with $\theta=3$ and do not admit UAM either.
\end{example}

\begin{remark}[Origin of the discrepancy]\label{rem:gmr-source}
In the proof of \cite[Lemma~5.2]{GMR} the computation for family (i) gives the even parity, while the paragraph on family (ii) states the odd parity for the graph and the even parity for its complement. These cannot both hold, because for $n=4\theta^2$ and $\displaystyle t=\frac{\pi}{2\theta}$ one has $e^{itJ}=I$ and hence $U_{\overline X}(t)=e^{-it}\,\overline{U_X(t)}$, which is the argument of \cite[Lemma~2.1]{GMR}. The families (iii), (iv) and (v) are treated in \cite[Sections~5.2 to~5.4]{GMR}, where the argument for (iii) is only sketched. Lemma~\ref{lem:no-single-time} and Propositions~\ref{prop:square-conference} and~\ref{prop:no-single-time} give independent proofs that these families admit no instantaneous uniform mixing except for the order nine.
\end{remark}

\begin{remark}[Two open problems of Coutinho and Guo]\label{rem:cg}
By Theorem~\ref{thm:gmr-corrected} and its proof, together with
Sections~\ref{sec:examples} and~\ref{sec:edge}, every time $t$
of instantaneous uniform mixing on a connected strongly regular
graph with at least two vertices is a rational multiple of $\pi$,
and all adjacency eigenvalues of such a graph are integers.
Consequently, all eigenvalues of $U(t)$ are roots of unity.
This answers \cite[Open Problem~3.1]{CG} within the class of
strongly regular graphs.

Together with Section~\ref{sec:edge}, the corrected classification
in Theorem~\ref{thm:gmr-corrected} also answers
\cite[Open Problem~3.3]{CG} within this class, by giving necessary
and sufficient conditions for instantaneous uniform mixing
in terms of the spectrum alone, without reference to a
specific mixing time.
\end{remark}

\section*{Disclosure of AI assistance}
Through a reference-checking process supported by the large language model Claude Opus 5.5, we found the need for the correction of Theorem~5.1 and Lemma~5.2 in \cite{GMR} given in Section~\ref{sec:gmr}, and the classification in Theorem~\ref{thm:classification}. The model also proposed the arguments of Lemma~\ref{lem:observable-compression}, Remarks~\ref{rem:two-scope} and~\ref{rem:conference-explicit} and Propositions~\ref{prop:elementary-chan}, \ref{prop:no-single-time} and~\ref{prop:density}. The author checked all contents, including Theorem~\ref{thm:classification}, and takes full responsibility for them.


\begin{thebibliography}{99}

\bibitem{ABTW}
A.~Ahmadi, R.~Belk, C.~Tamon, and C.~Wendler,
\emph{On mixing in continuous-time quantum walks on some circulant graphs},
Quantum Inf. Comput. \textbf{3}(6) (2003), 611--618.

\bibitem{BCM}
P.~Baptista, G.~Coutinho, and V.~Marques,
\emph{Unexpected averages of mixing matrices},
Quantum Information Processing \textbf{23}(7), article 272 (2024).
\url{https://doi.org/10.1007/s11128-024-04480-0}.
Preprint: \url{https://arxiv.org/abs/2308.16378}. Statement numbers refer to the arXiv version.

\bibitem{BK}
I.~B\'ar\'any and R.~Karasev,
\emph{Notes about the Carath\'eodory number},
Discrete \& Computational Geometry \textbf{48} (2012), 783--792.
See Section~1 and Theorem~1.5 for the Fenchel and
Hanner--R{\aa}dstr\"om refinements.
\url{https://arxiv.org/abs/1112.5942}.

\bibitem{BT}
C.~Bayer and J.~Teichmann,
\emph{The proof of Tchakaloff's theorem},
Proc. Amer. Math. Soc. \textbf{134}(10) (2006), 3035--3040.
\url{https://doi.org/10.1090/S0002-9939-06-08249-9}.

\bibitem{BKMT}
A.~Best, M.~Kliegl, S.~Mead-Gluchacki, and C.~Tamon,
\emph{Mixing of quantum walks on generalized hypercubes},
Int. J. Quantum Inf. \textbf{6}(6) (2008), 1135--1148.

\bibitem{BH}
A.~E.~Brouwer and W.~H.~Haemers,
\emph{Spectra of Graphs}, Universitext, Springer, 2012.
\url{https://doi.org/10.1007/978-1-4614-1939-6}.

\bibitem{Cameron}
P.~J.~Cameron,
\emph{Strongly regular graphs}, in L.~W.~Beineke and R.~J.~Wilson (eds.),
Topics in Algebraic Graph Theory, Cambridge University Press, 2004, 203--221.
Draft of April 2001:
\url{https://maths.qmul.ac.uk/~lsoicher/designtheory.org/library/preprints/srg.pdf}.

\bibitem{Chan11}
A.~Chan,
\emph{Complex Hadamard matrices and strongly regular graphs},
arXiv:1102.5601v2 (2020).

\bibitem{Chan20}
A.~Chan,
\emph{Complex Hadamard matrices, instantaneous uniform mixing and cubes},
Algebraic Combinatorics \textbf{3}(3) (2020), 757--774.
\url{https://doi.org/10.5802/alco.112}.

\bibitem{CG}
G.~Coutinho and K.~Guo,
\emph{Selected open problems in continuous-time quantum walks},
Special Matrices \textbf{12}(1) (2024), article 20240025.
\url{https://doi.org/10.1515/spma-2024-0025}.

\bibitem{GMR}
C.~Godsil, N.~Mullin, and A.~Roy,
\emph{Uniform mixing and association schemes},
Electronic Journal of Combinatorics \textbf{24}(3), P3.22 (2017).
\url{https://doi.org/10.37236/4745}.

\bibitem{Sch}
K.~Schm\"udgen,
\emph{The Moment Problem},
Graduate Texts in Mathematics 277, Springer, 2017.
\url{https://doi.org/10.1007/978-3-319-64546-9}.

\bibitem{Zhan}
H.~Zhan,
\emph{$\varepsilon$-uniform mixing in discrete quantum walks},
arXiv:2311.18797 (2023).

\end{thebibliography}
\end{document}